\documentclass[a4paper,reqno]{amsart}

\usepackage[margin=3cm]{geometry}
\usepackage{mathtools}
\usepackage{amsmath}
\usepackage{amssymb}
\usepackage{amsthm}
\usepackage{amsfonts}
\usepackage{thmtools,thm-restate}

\usepackage[
  colorlinks,
  linkcolor = blue,
  citecolor = blue,
  urlcolor = blue]{hyperref}

\usepackage[utf8]{inputenc}

\usepackage{physics}
\usepackage[export]{adjustbox}
\usepackage{tabularx}
\usepackage{float}

\usepackage{subfig}
\usepackage{multirow}
\graphicspath{{Figures/}{../Figures/}}
\usepackage{subfiles}
\usepackage{wrapfig}
\usepackage{titling}
\usepackage{appendix}
\usepackage{xcolor}
\usepackage{qcircuit}
\usepackage{algorithm}
\usepackage{bbm}
\usepackage{algpseudocode}
\usepackage{physics}
\usepackage{tikz}
\usepackage{pgfplots}
\pgfplotsset{compat=1.18}

\usepackage{enumitem} 

\usepackage{cleveref}

\usepackage[foot]{amsaddr}

\usepackage[backend=biber,style=alphabetic,maxnames=9,maxalphanames=5,isbn=false,backref]{biblatex} 
\makeatletter
\renewcommand\subsubsection{\@secnumfont}{\bfseries}%
\renewcommand\subsubsection{\@startsection{subsubsection}{3}
  \z@{.5\linespacing\@plus.7\linespacing}{-.5em}%
  {\normalfont\bfseries}}
  
\makeatother

\newtheorem{theorem}{Theorem}

\newtheorem{lemma}{Lemma}
\newtheorem{corollary}{Corollary}
\newtheorem{definition}{Definition}
\newtheorem{proposition}{Proposition}

\theoremstyle{definition}

\newtheorem{remark}{Remark}

\DeclarePairedDelimiter{\set}{\lbrace}{\rbrace}
\DeclarePairedDelimiter{\floor}{\lfloor}{\rfloor}

\DeclarePairedDelimiter{\of}{\lparen}{\rparen}

\DeclarePairedDelimiter{\ipp}{\langle}{\rangle}

\newcommand{\id}{\mathbbm{1}}
\newcommand{\pt}{\mathbin{\vdash}} 
\newcommand{\defeq}{\vcentcolon=}

\newcommand{\maj}{\mathrm{maj}}
\newcommand{\SYT}{\mathrm{SYT}}

\newcommand{\tp}{^{\mathsf{T}}}

\newcommand{\QSI}[1][]{\textup{QSI}}

\let\S\relax
\DeclareMathOperator{\S}{S} 
\DeclareMathOperator{\C}{\mathbb{C}} 
\DeclareMathOperator{\CS}{\mathbb{C}S} 
\DeclareMathOperator{\End}{End} 

\usepackage{mdframed}
\newcounter{custalgocounter}
\renewcommand{\thecustalgocounter}{\arabic{custalgocounter}}
\newenvironment{custalgo}[1][]
{%
  \noindent \begin{minipage}{\textwidth}
    \begin{mdframed}[
      linewidth=0.8pt,
      roundcorner=5pt,
      backgroundcolor=white!,
      innertopmargin=10pt,
      innerbottommargin=10pt,
      innerleftmargin=10pt,
      innerrightmargin=10pt,
      skipabove=\topsep,
      skipbelow=\topsep,
    ]
      \refstepcounter{custalgocounter}
      \begin{center}
        \textbf{Algorithm \thecustalgocounter:} #1
        \linebreak
      \end{center} 
}%
{%
    \end{mdframed}
      \end{minipage}
}

\title{Permutation tests for quantum state identity}

\author{Harry Buhrman$^{1,2,3}$}
\email{h.m.buhrman@uva.nl}
\address{\hspace{-21pt}$^1$QuSoft, Amsterdam, The Netherlands}
\address{$^2$Institute for Logic, Language and Computation, University of Amsterdam, The Netherlands}
\address{$^3$Quantinuum, Partnership House, Carlisle Place, London, United Kingdom}

\author{Dmitry Grinko$^{1,2}$}
\email{d.grinko@uva.nl}

\author{Philip Verduyn Lunel$^{1,4}$}
\email{philip.verduyn.lunel@cwi.nl}
\address{$^4$Centrum Wiskunde \& Informatica, Amsterdam, The Netherlands}

\author{Jordi Weggemans$^{1,4}$}
\email{jrw@cwi.nl}

\begin{document}

\begin{abstract}
The quantum analogue of the equality function, known as the \textit{quantum state identity problem}, is the task of deciding whether $n$ unknown quantum states are equal or unequal, given the promise that all states are either pairwise orthogonal or identical. Under the one-sided error requirement, it is known that the Permutation test is optimal for this task, and for two input states this coincides with the well-known Swap test. Until now, the optimal measurement in the general two-sided error regime was unknown. Under more specific promises, the problem can be solved approximately or even optimally with simpler tests, such as the Circle test.

This work attempts to capture the underlying structure of the quantum state identity problem. Using tools from semidefinite programming and representation theory, we (i) give an optimal test for any input distribution without the one-sided error requirement by writing the problem as an SDP, giving the exact solutions to the primal and dual programs and showing that the two values coincide; (ii) propose a general $G$-test which uses an arbitrary subgroup $G$ of $\S_n$, giving an analytic expression of the performance of the specific test, and (iii) give an approximation of the Permutation test using only a classical permutation and $n-1$ Swap tests. 
\end{abstract}
\maketitle

\tableofcontents

\section{Introduction}
\noindent
One of the most basic tasks in quantum information is to compare two or more unknown quantum states.
In particular, we will be interested in a quantum version of the equality function, known as the \textit{quantum state identity problem} ($\QSI_{n}$)~\cite{kada2008efficiency}.
Here one is given $n$ \textit{unknown} quantum states $\ket{\psi_1},\ket{\psi_2},\dotsc,\ket{\psi_n}$, each of local dimension $d$, with the promise that all states are either pairwise orthogonal or identical.
Given as an input such a collection of unknown states satisfying the promise, the task is now to decide whether the input states are all equal or not, with high probability.
This problem arises naturally in quantum fingerprinting~\cite{buhrman2001fingerprinting} and in verification tasks in quantum networks, where one wishes to certify that multiple quantum states were prepared identically without revealing any information about the states themselves.

For two input states, the corresponding problem $\QSI_2$ can be solved by the so-called \emph{Swap Test}~\cite{barenco1997stabilization,buhrman2001fingerprinting}.
The Swap test applies, controlled by a qubit in the $\ket{+}$ state in the first register, a controlled Swap operation to the input states, followed by a Hadamard operation and measurement in the computational basis of the first register.
The controlled swap operation can be viewed as applying in superposition all permutations of the symmetric group $\S_2$ to the input states, which consists of only the identity and swap of both input states.
It is known that under the so-called \emph{one-sided error requirement}, i.e.~when the algorithm has to be always correct on equal inputs, the Swap Test is optimal for $\QSI_{2}$~\cite{kobayashi2001quantum,de2004one}.
For arbitrary $n$, one can use generalizations of the Swap test like the circle test or the permutation test~\cite{barenco1997stabilization,buhrman2001fingerprinting,kada2008efficiency}.
Both tests extend the Swap test in the sense that they perform in superposition all permutations from a subgroup of the symmetric group $\S_n$: for the circle test this is the cyclic group $\textup{C}_n$ and for the permutation test the subgroup is $\S_n$ itself.
Kada, Nishimura and Yamakami proved the optimality of the permutation test for $\QSI_n$ for arbitrary $n$ under the one-sided error requirement~\cite{kada2008efficiency}.
It was not known whether both tests are still optimal once the one-sided error requirement is relaxed.

Even though the permutation test is optimal under the one-sided error requirement, its circuit complexity is considerable: it requires an ancilla register of dimension $n!$ and a Fourier transform over the symmetric group $\S_n$~\cite{bradshawCycleIndexPolynomials2023a,labordeTestingSymmetryQuantum2023a}.
In contrast,~\cite{kada2008efficiency} showed that when $n$ is prime, the circle test achieves the same optimal soundness error for $\QSI_n$, while only requiring a Fourier transform over the cyclic group $\textup{C}_n$ on a register of dimension $n$ --- an exponential reduction.
This suggests that even though the permutation test is \textit{information-theoretically} optimal, there might be specific instances of $\QSI_{n}$ for which other tests have better \textit{circuit complexity} whilst still being optimal.
Something similar was observed in~\cite{Chabaud2018optimal}, where a Swap test-based circuit was shown to be optimal for $\QSI_n$ when the promise is that exactly one state is different and its location is known.

Moreover, if one slightly relaxes the requirement of optimality,~\cite{kada2008efficiency} showed that when one randomly permutes the input before the circle test is performed, it approximates the permutation test in the sense that it retains perfect completeness and the dominant term in the soundness error achieves the same scaling up to a multiplicative factor smaller than $2$, which holds for arbitrary $n$ instead of just prime numbers.
They also gave an approximation procedure based on the Swap test to approximate $\QSI_3$, for which both the Circle test and Permutation test are optimal.

These results hint at some underlying mathematical structure in $\QSI_n$, which is currently not fully understood.
Using tools from representation theory and semidefinite programming and by fine-graining the promises of the quantum state identity problem, this work seeks to unravel this structure, allowing us to answer some of the 15-year-old open questions related to the quantum state identity problem.
Concretely, we prove that the permutation test is optimal for all formulations of the quantum state identity problem even when the one-sided error requirement is relaxed, we derive an exact formula for the performance of any $G$-test for an arbitrary subgroup $G \subseteq \S_n$ in terms of Kostka numbers and representation-theoretic multiplicities, and we introduce the \emph{Iterated Swap Tree} as a new protocol that achieves near-optimal soundness using only classical randomness and $n-1$ Swap tests.

\section{Results}
Write $\mu \pt n$ for a partition of the number $n$, that is $\mu = (\mu_1,\dots,\mu_d)$ such that $\mu_1 \geq \dotsc \geq \mu_d \geq 0$ and $\sum_{i=1}^d \mu_i = n$. We also write $\lambda \pt_d n$ to mean that $\lambda$ is a partition of at most $d$ parts, i.e. for every $i>d$ we have $\lambda_i = 0$. We associate with $\mu$ corresponding quantum state $\ket{1^{\mu_1} 2^{\mu_2} \dotsc d^{\mu_d}}$, where $\ket{i^{\mu_i}} \defeq \ket{i}^{\otimes \mu_i}$ for a basis state $\ket{i}$ from a Hilbert space of local dimension $d$. Note that the number of non-zero entries in $\mu$ specifies the total number of different basis states in the input state. We can now define several formulations of the quantum state identity problem.

\begin{definition}[Quantum State Identity problems] \label{def:QSI}
\phantom{} Let  $d \in \mathbb{N}$, $p \in [0,1]$. We consider the following promise problems given by the following input, promise and output: \\
\textbf{Input:} $n$ unknown quantum states $\ket{\psi_1},\ket{\psi_2},\dotsc,\ket{\psi_n}$, each of local dimension $d$ and pairwise orthogonal or identical.\\
\textbf{Promise:} We have that either
\begin{enumerate}[label=(\roman*)]
    \item All $\ket{\psi_i}$ are identical.
    \item We have one of the promises for each of the following problems:
    \begin{enumerate}
        \item[$\QSI_\mu^{p}$:] One is given some $\mu \pt n$  with $\mu_2>0$ such that for some unknown unitary $U$ of dimension $d$
        \begin{equation*}
            \ket{\psi_1}\ket{\psi_{2}}\dotsc\ket{\psi_{n}} = U^{\otimes n} \ket{1^{\mu_1} 2^{\mu_2} \dotsc d^{\mu_d}},
        \end{equation*}
        \item[$\widetilde{\QSI}_\mu^{p}$:] One is given some $\mu \pt n$ with $\mu_2>0$ such that for some unknown unitary $U$ of dimension $d$ and some unknown permutation $\sigma$
        \begin{equation*}
            \ket{\psi_{\sigma(1)}}\ket{\psi_{\sigma(2)}}\dotsc\ket{\psi_{\sigma(n)}} = U^{\otimes n} \ket{1^{\mu_1} 2^{\mu_2} \dotsc d^{\mu_d}}.
        \end{equation*}
        \item[${\QSI}_n^{p}$:] There exist an $i,j \in [n]$ such that $\ket{\psi_i}$, $\ket{\psi_j}$ are pairwise orthogonal.
    \end{enumerate}
\end{enumerate}
Moreover, one is promised that (i) happens with probability $p$ and (ii) with $1-p$. If $p$ is not specified (denoted ${\QSI_\mu}$), one can assume that $p=1/2$.\\
\textbf{Output:} Return `\emph{equal}' in case (i) and `\emph{not equal}' in case (ii). 
\end{definition}
Similar to the type I and type II errors in hypothesis testing, we refer to the probability of being correct in case (i) as the \emph{completeness} probability and in case (ii) as the \emph{soundness} probability. 
The overall success probability is then the average over both the completeness and soundness cases, weighted according to the prior $\{p,1-p\}$. 
When we say perfect completeness (or soundness), we mean that the probability of being correct in this case is $1$.

Note that the quantum state identity problems are stated in such a way that the promise in each formulation captures less information as we go from top to bottom. $\hyperref[def:QSI]{\QSI_\mu^{p}}$ captures the setting where one has all the information on where different states are located, as one can simply relabel the input states to obtain this form. 
This captures, for example, the setting of~\cite{Chabaud2018optimal} where $\mu = (n-1,1,0,\dots,0)$. In $\hyperref[def:QSI]{\widetilde{\QSI}_\mu^{p}}$, all the information one has is in knowing how many states $h$ out of $n$ are different, since this is given by the number of non-zero entries in $\mu$. 
Finally, $\hyperref[def:QSI]{{\QSI}_n}$ is the usually adopted definition of the quantum state identity problem as given in~\cite{kada2008efficiency}, where in case (ii) one knows nothing of the location, nor how many of the states are different, except that there is at least one.

As an interesting comparison, in \cite{buhrman2022quantum} the authors give an optimal way to compute any symmetric, equivariant Boolean function in an unknown quantum basis, such as the majority and parity function. When we restrict ourselves to the qubit case input, we see that what we actually compute is the $\textup{OR}$-function, which is not an equivariant Boolean function, in an unknown quantum basis. 

\subsubsection*{Optimal test}
Our first result is a very general characterization of the optimality of the Permutation Test for quantum state identity problems.

\begin{theorem}[Informal, from Theorem~\ref{thm:permtest_trivial_optimal}] 
Suppose that the unequal and equal cases of the quantum state identity problems, as per Definition~\ref{def:QSI}, occur equally likely. Then for any $\mu \pt n$ (known or unknown), the permutation test is optimal for all quantum state identity problems, satisfies perfect completeness, and has soundness 
\begin{align*}
    1-\frac{1}{\binom{n}{\mu}},
\end{align*}
where $\binom{n}{\mu} \defeq \frac{n!}{\mu_1!\dotsc\mu_d!}$ is a multinomial coefficient.
\label{thm:inf_PT_optimal}    
\end{theorem}

Several interesting conclusions can be drawn from Theorem~\ref{thm:inf_PT_optimal}:
\begin{enumerate}
    \item Relaxing the one-sided error requirement does not make the average success probability (averaged over the unequal and equal cases) higher, answering an open question by~\cite{kada2008efficiency}.
    \item Similarly, knowing where the different states are located also does not increase one's capacity to increase the overall success probability.
    \item The permutation test will on most inputs perform much better than the worst-case instances that have $1/n$ soundness error, which are the inputs where only one state is different from all the others.
\end{enumerate}
If $p \neq \frac{1}{2}$, then the permutation test is still optimal for all $p \geq p^*(\mu)$, where $p^*(\mu)\defeq\frac{1}{1+\binom{n}{\mu}}$. For all $p <p^*(\mu)$ the optimal test is to always output `not equal' (Theorem~\ref{thm:permtest_trivial_optimal}).

Our proof is based on an SDP formulation of finding the measurement that maximizes the success probability given the `equal' case with probability $p$. Our SDP considers input states that have been acted upon by a Haar random unitary. This is to some extent too strong an assumption: not only is the basis unknown to the problem solver, the basis is also `unknown to the universe'! In an adversarial setting, where an all-powerful entity could pick any measure over the unitary group to be applied to the input states, the input states would still be unknown but the Haar measure might not be the hardest input distribution to solve. However, we show that this assumption can be made without loss of generality, exploiting a special property of the structure of the permutation test which is the optimal solution to our SDP (Subsection~\ref{subsec:unknown_haar_minimax}).

\subsubsection*{\texorpdfstring{The $G$-test}{The G-test}}
Whilst Theorem~\ref{thm:inf_PT_optimal} fully resolves the quantum state identity problem in terms of determining the optimal information-theoretic measurement, it might for some be still `too blunt of a tool' as its circuit complexity is considerable. In that case, one might be able to resort to other tests with smaller circuit complexity because (i) more is known about the type of inputs or (ii) one is only interested in a sub-optimal performance guarantee. For example, the circle test is optimal for $n$ is prime and a good approximation for arbitrary $n$ when the circle test is preceded by a random (classical) permutation of the input states~\cite{kada2008efficiency}.

Using these ideas, we define a general notion of a so-called \textit{$G$-test}, which does the following.
\begin{enumerate}
    \item Randomly permute the labels of the input states.
    \item Apply the permutation test for a subgroup $G \subseteq \S_n$ (see Algorithm~\ref{alg:perm_G_test}).
\end{enumerate}

\begin{theorem}
\label{thm:Gtest}
Let $G$ be any subgroup of $\S_n$. Then the $G$-test has perfect completeness and soundness
\begin{equation*}
    \mathbb{P}^{G}_{s}(\mu) = 1 - \frac{1}{\binom{n}{\mu}} \sum_{\substack{\lambda \pt_d n}} K_{\lambda,\mu} r^{G}_{\lambda},
\end{equation*}
where $K_{\lambda,\mu}$ are Kostka numbers, $r^{G}_{\lambda}$ is the multiplicity of the trivial irrep of the subgroup $G \subseteq \S_n$ inside the irrep $\lambda$ of the symmetric group $\S_n$. 
\end{theorem}
The derivation of Theorem~\ref{thm:Gtest} can be found in Section~\ref{sec:G_test_perf}, where we also briefly touch upon the cases of $G$ being the cyclic group $\textup{C}_n$ (which captures the RCIR test of~\cite{kada2008efficiency}) as well as an iterated wreath product of the symmetric groups, see Section~\ref{sec:def_wreath_test}.

\subsubsection*{Approximation by the Iterated Swap Tree}
Whilst Theorem~\ref{thm:Gtest} in principle contains all the information to compute the performance of any $G$-test, this might in practice be quite laborious for most subgroups $G$. In particular, we will also consider a test that consists solely of a classical permutation followed by Swap tests in a sort of tree structure, which is captured by the iterated wreath product~\cite{orellana2004rooted,im2018generalized}. This test, which we call the \emph{Iterated Swap Tree} (\hyperref[alg:IST]{IST}), works on any input size that is a power of two and only uses a random classical permutation and $n-1$ Swap tests. It can be viewed as a generalization of the test as proposed in~\cite{Chabaud2018optimal}, with the extra advantage that it no longer needs to have in the unequal case a promise that the different state is located in the first register. 

Using different techniques and ideas, we provide a recursive upper bound on the error probability as a function of the number of different states $h$ and the total number of input states $n$.

\begin{theorem}[Informal, from Theorem~\ref{thm:swap_trees} and Corollary~\ref{cor:IST}]
\label{thm:IST}
For any $n$ that is a power of two the Iterated Swap Tree achieves perfect completeness and has soundness lower bounded by
    \begin{equation}
        \mathbb{P}^{\mathrm{IST}}_{s}(n,h) \geq 1-  \frac{\gamma(h,\log_2(n))}{\binom{n}{h}},
    \end{equation}
    where $\gamma(h,m)$ satisfies a recurrence relation with the boundary conditions $\gamma(0,m) = \gamma(1,m) = 1$. Moreover, for any fixed $h$ we have that $\gamma(h,\log_2(n))$ is a polynomial in $\log_2(n)$ of degree $h-1$, which implies that there exists a $n_0 \in \mathbb{N}$ such that for all $n \geq n_0$ the test achieves the optimal soundness  $\geq 1- 1/n$.
\end{theorem}

Our proof is based on the combinatorics of what we call `clicks', which are potential detection locations where the inputs to a Swap test within the Iterated Swap Tree have the possibility of being detected. Since for unequal input states this happens with probability $\frac{1}{2}$, the problem boils down to computing the expectation value $\mathbb{E}_\sigma[\left( \frac{1}{2} \right)^{c(h,m,\sigma)}]$, where $c(h,m,\sigma)$ is a function that gives the total number of clicks given $n = 2^m$ input states of which $h$ are different and a permutation $\sigma$. The fact that we take the expectation value over all permutations comes from the random permutation of the labels we perform, and this is necessary since one can easily show that the test would otherwise have large soundness error on some specific instances (in this case, it would be to have the first and second half of the states different, but identical within the halves themselves).

\section{Preliminaries}
\subsection{Notation}

Let $n \geq 0$ be an integer.
We write $\lambda \pt n$ to mean that $\lambda$ is a \emph{partition} of $n$, i.e.,
$\lambda = (\lambda_1,\dotsc,\lambda_d)$ for some $d$ where
$\lambda_1 + \dotsb + \lambda_d = n$ and
$\lambda_1 \geq \dots \geq \lambda_d \geq 0$.
The \emph{size} of partition $\lambda$ is $\abs{\lambda} = n$.
We say that $\lambda$ has \emph{length} $\ell(\lambda) = k$ if $\lambda_k > 0$ and $\lambda_{k+1} = 0$. We use the notation $\lambda \pt_d n$ to indicate that $\lambda \pt n$ and $\ell(\lambda) \leq d$.
The \emph{Young diagram} of partition $\lambda$ is a collection of $n$ \emph{cells} arranged in $\ell(\lambda)$ rows, with $\lambda_i$ of them in the $i$-th row.

\subsection{Representation theory background}

Broadly speaking, representation theory studies how one can represent general algebraic objects such as groups, associative algebras and Lie algebras as linear operators with matrix multiplication preserving the original algebraic operations. The fundamental concept here is of \emph{irreducible representation} (also simply called "irrep"), i.e.~a representation (also called "module") which does not have (as a vector space) any invariant subspaces under the action of the represented group or algebra. When studying the actions of general groups on a given vector space it is useful to understand how this representation decomposes into irreducible representations of the given group and other relevant algebraic structures.

A canonical example of such an approach is \emph{Schur--Weyl duality} which describes the following situation. Suppose that you study the diagonal action of the unitary group $\mathrm{U}_d$ on the vector space $(\C^d)^{\otimes n}$. It turns out that the commutant of this action is spanned by permutation operators $\sigma$ from the symmetric group $\S_n$ and vice versa: the commutant of $\S_n$ in $(\C^d)^{\otimes n}$ is spanned by $U^{\otimes n}$ where $U \in \mathrm{U}_d$. In turn, we can decompose $(\C^d)^{\otimes n}$ into irreducible representations of unitary and symmetric groups:
\begin{equation}\label{eq:schur_weyl_decomp}
    (\C^d)^{\otimes n} \cong \bigoplus_{\lambda \pt_d n} W_\lambda \otimes V_\lambda,
\end{equation}
where $W_\lambda$ is so-called \emph{Weyl module}, i.e.~an irreducible representation of $\mathrm{U}_d$, and $V_\lambda$ is so-called \emph{Specht module}, i.e.~an irreducible representation of $\S_n$. The dimensions of these irreps we denote by 
\begin{equation}
    d_\lambda \defeq \dim(V_\lambda), \qquad m_{\lambda,d} \defeq \dim(W_\lambda).
\end{equation}
The isomorphism in \cref{eq:schur_weyl_decomp} can be achieved via a unitary $U_{\mathrm{Sch}}$ called \emph{Schur transform} \cite{bch2006quantumschur,HarrowThesis,kirby2018practical,Krovi2019,wills2023generalised,grinko2023gelfandtsetlin,nguyen2023mixed}. In particular, we can introduce the \emph{isotypic projectors} $\Pi_\lambda$ as follows:
\begin{equation}
   U_{\mathrm{Sch}} \Pi_\lambda U_{\mathrm{Sch}}^{\dagger} = \bigoplus_{\mu \pt_d n} \delta_{\mu,\lambda} I_\mu \otimes I_\mu.
\end{equation}

\subsubsection{Weingarten calculus}

One crucial tool which we need is \emph{Weingarten calculus} \cite{collins2006integration,Collins_2022}. It is a set of results on integration with respect to the Haar measure of the unitary group. In particular, we need the following result \cite[Proposition~2.3]{collins2006integration} regarding an integral of an arbitrary matrix $X \in \End((\C^d)^{\otimes n})$:
\begin{equation}
    \int U^{\otimes n} X U^{\dagger \otimes n} dU = \of*{\sum_{\sigma \in \S_n} \Tr[ X \psi(\sigma^{-1})] \psi(\sigma) } \of[\Bigg]{ \frac{1}{n!} \sum_{\substack{\lambda \pt_d n}} \frac{d_\lambda}{m_{\lambda,d}} \Pi_\lambda },
\end{equation}
where we formally defined an action of $\S_n$ on $(\C^d)^{\otimes n}$ which permutes the registers as the representation $\psi : \S_n \rightarrow \End((\C^d)^{\otimes n})$.

\subsection{Semidefinite programming}

\emph{Semidefinite programming} is an important subfield of optimization \cite{handbookSDP} that has numerous applications in quantum information theory \cite{SDPs,watrous}. A typical formulation of a \emph{semidefinite program} (SDP) has the form \cite[Section~1.1]{handbookSDP}:
\begin{equation}
  \begin{aligned}
    f \defeq \max_X \quad & \Tr \of{C\tp X} \quad \textrm{s.t.} \quad \Tr \of{A_i\tp X} = b_i \quad \forall i \in [m], \quad X \succeq 0,
  \end{aligned}
\end{equation}
where $X$ is a Hermitian matrix variable,
$C$ and $A_i$ are constant Hermitian matrices, and
$b_i$ are real constants. For this \emph{primal} semidefinite program there always exists the \emph{dual} semidefinite program:
\begin{equation}
  \begin{aligned}
    f^* \defeq \min_{y \in \mathbb{R}^m} \quad & y\tp b \quad \textrm{s.t.} \quad \sum_{i=1}^m y_i A_i \succeq C.
  \end{aligned}
\end{equation}
There is an important property called \emph{weak duality} which always holds: 
\begin{equation}
    f^* \geq f.
\end{equation}
This property allows to prove the optimality of a candidate feasible solution of the primal SDP by exhibiting a feasible solution of the dual program which matches the value of the primal objective function. This method of proving optimality is used in \Cref{sec:perm_test}.

\section{\texorpdfstring{The $G$-test}{The G-test}}

In this section, we present a generalization of the Swap test to arbitrary subgroups of the symmetric group. For a subgroup $G \subseteq \S_n$, let $F_G$ be a $\abs{G} \times \abs{G}$ unitary operator, which prepares an equal superposition of all elements of the group:
\begin{equation}
    F_G \ket{0} = \ket{+}, \qquad \ket{+} \defeq \frac{1}{\sqrt{\abs{G}}}\sum_{g \in G} \ket{g}.
\end{equation}
We also define the controlled permutation operation $C$ as the map that implements
\begin{align}
  C : \ket{\sigma} \ket{\psi_1} \ket{\psi_2} \dots \ket{\psi_n} \mapsto \ket{\sigma} \ket{\psi_{\sigma^{-1}(1)}} \ket{\psi_{\sigma^{-1}(2)}} \dots \ket{\psi_{\sigma^{-1}(n)}}.
\end{align}
Then we can describe a general algorithm for the quantum state identity problem as follows:

\begin{custalgo}[Permutation test for a subgroup $G \subseteq \S_n$ ($G$-test)] \label{alg:permtest_subgroup}
\noindent \textbf{Input}: $n$ quantum states $\ket{\psi_1} \otimes \ket{\psi_2} \otimes \dots \otimes \ket{\psi_n}$, each of local dimension $d$.\\
\begin{enumerate}
    \item Start with the initial state $\ket{0} \ket{\psi_1} \ket{\psi_2} \dots \ket{\psi_n}$,
    where we denote the $|G|$-dimensional register containing $\ket{0}$ as $P$ and the register with the input states as $R$.
    \item Apply $F_G$ to register $P$.
    \item Apply $C$ to register $R$, controlled by register $P$. 
    \item Apply $F_G^\dagger$ to $P$ and measure register $P$. If the outcome is $0$, return \textit{equal}, otherwise return \textit{not equal}.
\end{enumerate}
\label{alg:perm_G_test}
\end{custalgo}

If $G = \S_n$ then the corresponding test is known as the \emph{permutation test}, for $G = \textup{C}_n$ the \emph{circle test}~\cite{kada2008efficiency} and for $G = \S_2$ the \emph{Swap test}~\cite{buhrman2001fingerprinting}. If the final measurement outcome is $\ket{0}$, the circuit has projected the input state onto the subspace containing all permutations $\sigma \in G$, a technique that was first proposed in~\cite{barenco1997stabilization}. 
The explicit description of quantum circuits for the permutation and circle tests was provided in \cite{bradshawCycleIndexPolynomials2023a}, which builds on the results of \cite{labordeTestingSymmetryQuantum2023a}. 

In the next subsection, we will focus solely on $G = \S_n$, and come back to the general setting of an arbitrary subgroup $G$ in Section~\ref{sec:G_test_perf}.

\subsection{The Permutation test}
\label{sec:perm_test}

In this section, we study the Permutation test, i.e. $G = \S_n$. We approach the problem using techniques from semidefinite optimization and representation theory.

\subsubsection{The SDP formulation of the optimal measurement for Haar-random inputs}
To prove that in general that the permutation test gives the best probability of success for pure state equality testing we write the problem of finding the optimal measurement as an SDP. We consider two cases of the inputs. One where the order of the input in the unequal case is known, and one where the inputs are randomly permuted. We then show that there exists a feasible solution to the dual that coincides with the success probability of the permutation test for the case where the input order is known. Furthermore, we see that the permutation test is agnostic about the order of the inputs and achieves the same success probability on the randomly permuted input as on the ordered input. Therefore knowing the ordering in advance does not increase the optimal success probability. We define these different cases for equal or unequal qudit inputs as follows:
\begin{align}
    \rho_{\neq}^{\mu} &\defeq \int U^{\otimes n} \ket{1^{\mu_1} 2^{\mu_2} \dotsc d^{\mu_d}} \bra{1^{\mu_1} 2^{\mu_2} \dotsc d^{\mu_d}} (U^\dagger)^{\otimes n} dU, \label{def:rho_neq_mu} \\
    \tilde{\rho}_{\neq}^{\mu} &\defeq \frac{1}{n!}\sum_{\pi \in \S_n} \psi(\pi)\rho_{\neq}^{\mu} \psi(\pi^{-1}), \label{def:tilde_rho_neq_mu} \\
    \rho_{=}^{n} &\defeq \int U^{\otimes n} \ket{1^n} \bra{1^n} (U^\dagger)^{\otimes n} dU, \label{def:tilde_rho_n}
\end{align}
where $\mu \pt n$ is some partition of $n$ with $\mu_2>0$,\footnote{The condition $\mu_2>0$ is assumed automatically everywhere in this paper for the state $\rho^{\mu}_{\neq}$. On the other hand, the state $\rho^{n}_{=}$ corresponds to $\mu = (n)$. Therefore, one can think about subscripts as follows: "$=$" means $\mu_2=0$, while "$\neq$" means $\mu_2 > 0$.} and we integrate with respect to the Haar measure on the unitary group $\mathbb{U}(d)$.\footnote{Without loss of generality, we consider the Haar-random unitaries. The reason for that is explained in \Cref{subsec:unknown_haar_minimax}.} The task is now to find a POVM set that discriminates optimally between the two. We can exactly find the optimal discrimination using an SDP \cite{boyd2004convex}. The SDP that optimizes over all POVMs for the ordered inputs can be written as follows:

\begin{minipage}{0.45\textwidth}
\begin{align*}
    &\textbf{Primal Program} \\
    \textbf{maximize: } &p \Tr[ \Pi_= \rho_{=}^{n}] + (1-p) \Tr[ \Pi_{\neq} \rho_{\neq}^{\mu} ] \\
    \textbf{subject to: } &\Pi_= + \Pi_{\neq} = \id_{d^{n}} \\
    &\Pi_= \succeq 0, \\
    &\Pi_{\neq} \succeq 0.\\
\end{align*}
\end{minipage}
\begin{minipage}{0.5\textwidth}
\begin{align} \label{def:SDP}
    &\textbf{Dual Program} \nonumber \\
    \textbf{minimize: } &\Tr[Y] \nonumber \\
    \textbf{subject to: } &Y - p \rho_{=}^{n} \succeq 0,\\
    &Y - (1 - p)\rho_{\neq}^{\mu} \succeq 0,\nonumber \\
    &Y \in \text{Herm}((\mathbb{C}^d)^{\otimes n}) .\nonumber \\ \nonumber
\end{align}
\end{minipage}

\noindent To solve the above SDP we would need the following crucial lemma.

\begin{lemma} \label{lem:main_lemma}
Let $\mu = (\mu_1, \dotsc, \mu_d)$ be some partition of $n$. Define 
    \begin{align*}
        \rho_{\neq}^{\mu} &\defeq \int U^{\otimes n} \ket{1^{\mu_1} 2^{\mu_2} \dotsc d^{\mu_d}} \bra{1^{\mu_1} 2^{\mu_2} \dotsc d^{\mu_d}} (U^\dagger)^{\otimes n} dU, \qquad
        \tilde{\rho}_{\neq}^{\mu} \defeq \frac{1}{n!}\sum_{\pi \in \S_n} \psi(\pi)\rho_{\neq}^{\mu} \psi(\pi^{-1}).
    \end{align*}
    Then
    \begin{align}
        \rho_{\neq}^{\mu}
        &= \frac{1}{\binom{n}{\mu}} \sum_{\substack{\lambda \pt_d n}} \frac{d_\lambda}{m_{\lambda,d}} \of*{\Pi_{(\mu_1)} \otimes \dotsc \otimes \Pi_{(\mu_d)}} \Pi_\lambda, \qquad
        \tilde{\rho}_{\neq}^{\mu} = \frac{1}{\binom{n}{\mu}} \sum_{\substack{\lambda \pt_d n}} \frac{K_{\lambda \mu}}{m_{\lambda,d}} \Pi_\lambda,
    \end{align}
    where $\binom{n}{\mu} \defeq \frac{n!}{\mu_1!\dotsc\mu_d!}$ is a multinomial coefficient, $K_{\lambda \mu}$ is the Kostka number and $\Pi_\lambda \in \End \of{(\C^d)^{\otimes n}}$ is the isotypic projector onto $\lambda$ irrep in the tensor representation of the symmetric group $\S_{n}$. In particular, $\Pi_{(n)} = \frac{1}{n!}\sum_{\pi \in \S_{n}} \psi(\pi)$ and $\Pi_{(\mu_i)} = \frac{1}{\mu_i!}\sum_{\pi \in \S_{\mu_i}} \psi(\pi)$, where $\S_{\mu_i}$ is understood as a subgroup of $\S_{n}$ acting on the corresponding systems from the partition $\mu$ of $n$. In particular, for $\mu = (n)$ we have:
    \begin{equation}
        \rho_{=}^{n} \defeq \rho_{\neq}^{(n)} = \tilde{\rho}_{\neq}^{(n)} = \frac{\Pi_{(n)}}{\binom{n+d-1}{n}} ,
    \end{equation}
\end{lemma}
\begin{proof} 
Let $\chi^\lambda$ be the character of the irrep $\lambda$ of the symmetric group $\S_n$ and $s_{\lambda}$ is the Schur polynomial in $d$ variables, i.e. it is the character of the unitary group $
\mathbb{U}(d)$. The dimensions of the irreps of symmetric and unitary groups are $d_{\lambda} = \chi^\lambda(e)$ and $m_{\lambda,d} = s_{\lambda}(1,\dotsc,1)$. Let $\psi : \CS_n \rightarrow \mathrm{End}\of{(\C^d)^{\otimes n}}$ be the tensor representation of the symmetric group $\S_n$ acting on $(\C^d)^{\otimes n}$. Using the Weingarten calculus \cite[Proposition~2.3]{collins2006integration} we get the following:
    \begin{align}
        \rho_{\neq}^{\mu} &= \of*{\sum_{\sigma \in \S_n} \Tr[ \ket{1^{\mu_1} \dotsc d^{\mu_d}}\bra{1^{\mu_1} \dotsc d^{\mu_d}} \psi(\sigma^{-1})] \psi(\sigma) } \of[\Bigg]{ \frac{1}{n!} \sum_{\substack{\lambda \pt_d n}} \frac{d_\lambda}{m_{\lambda,d}} \Pi_\lambda } \\
        &= \of[\Bigg]{\sum_{\sigma \in S_{\mu_1} \times \dotsc \times S_{\mu_d}} \psi(\sigma) } \of[\Bigg]{ \frac{1}{n!} \sum_{\substack{\lambda \pt_d n}} \frac{d_\lambda}{m_{\lambda,d}} \Pi_\lambda }
        \\
        &=  \of[\bigg]{ \mu_1!\dotsc\mu_d! \, \Pi_{(\mu_1)} \otimes \dotsc \otimes \Pi_{(\mu_d)} } \of[\Bigg]{ \frac{1}{n!} \sum_{\substack{\lambda \pt_d n}} \frac{d_\lambda}{m_{\lambda,d}} \Pi_\lambda } \\
        &= \frac{\mu_1!\dotsc\mu_d!}{n!} \sum_{\substack{\lambda \pt_d n}} \frac{d_\lambda}{m_{\lambda,d}} \of*{\Pi_{(\mu_1)} \otimes \dotsc \otimes \Pi_{(\mu_d)}} \Pi_\lambda
    \end{align}
Therefore, the twirling of $\rho_{\neq}^{\mu}$ with respect to the symmetric group $\S_n$ is
    \begin{align}
        \tilde{\rho}_{\neq}^{\mu} &= \frac{1}{n!}\sum_{\pi \in \S_n} \psi(\pi) \rho_{\neq}^{\mu} \psi(\pi^{-1}) = \frac{\mu_1!\dotsc\mu_d!}{n!} \sum_{\substack{\lambda \pt_d n}} \frac{K_{\lambda \mu}}{m_{\lambda,d}} \Pi_\lambda,
    \end{align}
where we used Schur's lemma, the definition of the \emph{permutation module} and \emph{Kostka numbers} $K_{\lambda \mu}$ coming from the so-called Young’s Rule \cite{Sagan}. In particular, for $\mu = (n)$ we have
\begin{equation}
    \rho_{=}^{n} = \rho_{\neq}^{(n)} = \tilde{\rho}_{\neq}^{(n)} = \frac{\Pi_{(n)}}{\binom{n+d-1}{n}} ,
\end{equation}
where $\Pi_{(n)}$ is the projector onto the symmetric subspace of $n$ qudits, $(n)$ denotes a Young diagram corresponding to the symmetric subspace and $m_{(n),d} = \binom{n+d-1}{n}$ \cite{harrow2013church}. 
\end{proof}

\subsubsection{Unknown is the same as Haar-random}
\label{subsec:unknown_haar_minimax}
Previous section formulates discrimination between $\rho^{n}_{=}$ and $\rho^{\mu}_{\neq}$ via the SDP~\eqref{def:SDP},
where $\rho^{\mu}_{\neq}$ and $\rho^{n}_{=}$ are Haar-twirled states, cf.~\cref{def:rho_neq_mu,def:tilde_rho_neq_mu,def:tilde_rho_n}.
In the definition of
$\QSI^p_{\mu}$, however, the unitary $U$ is only \emph{unknown}, meaning that an adversary may choose an arbitrary probability measure $q$ over $U(d)$ and the tester must choose a measurement without knowing $q$.

To model this, consider POVM $\{\Pi_{=},\Pi_{\neq}=\id-\Pi_{=}\}$ and adversary measure $q$, with objective function to be optimized
\begin{equation}
g(\Pi_{=},q)
\defeq p\,\Tr[\Pi_{=}\rho^{n,(q)}_{=}] + (1-p)\,\Tr[(\id-\Pi_{=})\rho^{\mu,(q)}_{\neq}],
\end{equation}
where 
\begin{align}
    \rho^{n,(q)}_{=} &\defeq \int U^{\otimes n}|1^n\rangle\!\langle 1^n|(U^\dagger)^{\otimes n}\,dq(U),\\
    \rho^{\mu,(q)}_{\neq} &\defeq \int U^{\otimes n}|1^{\mu_1}2^{\mu_2}\cdots d^{\mu_d}\rangle\!\langle1^{\mu_1}2^{\mu_2}\cdots d^{\mu_d}|(U^\dagger)^{\otimes n}\,dq(U).
\end{align}
Let $X\defeq\{\Pi_{=} : 0\preceq \Pi_{=}\preceq \id\}$ and let $Q$ be the set of Borel probability measures on $U(d)$.
Then $X$ is convex and compact, $Q$ is convex and weak-$*$ compact. Moreover, $g(\Pi_{=},q)$ is affine in each argument: in particular, it is concave in $\Pi_{=}$ for each fixed $q$ and convex in $q$ for each fixed $\Pi_{=}$. Therefore, by von Neumann--Sion minimax theorem \cite{sion1958},
\begin{equation}\label{eq:sion_minimax}
\sup_{\Pi_{=}\in X}\inf_{q\in Q} g(\Pi_{=},q) \;=\; \inf_{q\in Q}\sup_{\Pi_{=}\in X} g(\Pi_{=},q).
\end{equation}

We now identify a least favourable $q$ using symmetry. For any $\Pi_{=}$ define its Haar-twirl
\begin{equation}
\overline{\Pi}_{=}\defeq\int_{V\sim \mathrm{Haar}} V^{\otimes n}\Pi_{=}V^{\otimes n\dagger}\,dV,
\qquad \overline{\Pi}_{\neq}\defeq I -\overline{\Pi}_{=}.
\end{equation}
For Haar-twirled hypotheses, $\rho^{n,(\mathrm{Haar})}_{=}$ and $\rho^{\mu,(\mathrm{Haar})}_{\neq}$ commute with all $V^{\otimes n}$,
hence $g(\overline{\Pi}_{=},q_{\mathrm{Haar}})=g(\Pi_{=},q_{\mathrm{Haar}})$. Thus there exists an optimal tester for
$q_{\mathrm{Haar}}$ that is $U(d)$-invariant, i.e. $V^{\otimes n}\Pi^\star_{=}V^{\otimes n\dagger}=\Pi^\star_{=}$ for all $V$.

For such an invariant $\Pi^\star_{=}$, for any fixed pure state $\ketbra{\phi}$ the quantity
$\Tr[\Pi^\star_{=}\,U^{\otimes n}\ketbra{\phi}(U^\dagger)^{\otimes n}]$ is independent of $U$ (by the invariance $V^{\otimes n}\Pi^\star_{=}V^{\otimes n\dagger}=\Pi^\star_{=}$), so averaging over any measure $q$
yields $g(\Pi^\star_{=},q)=g(\Pi^\star_{=},q_{\mathrm{Haar}})$ for all $q$. Consequently, using that $\Pi^\star_{=}$ is optimal for $q_{\mathrm{Haar}}$ by construction,
\begin{equation}
\inf_{q\in Q}\sup_{\Pi_{=}\in X} g(\Pi_{=},q)
\;\ge\; \inf_{q\in Q} g(\Pi^\star_{=},q)
\;=\; g(\Pi^\star_{=},q_{\mathrm{Haar}})
\;=\; \sup_{\Pi_{=}\in X} g(\Pi_{=},q_{\mathrm{Haar}}),
\end{equation}
while the reverse inequality holds trivially by choosing $q=q_{\mathrm{Haar}}$. Hence
\begin{equation}\label{eq:haar_least_favourable}
\inf_{q\in Q}\sup_{\Pi_{=}\in X} g(\Pi_{=},q)
\;=\;
\sup_{\Pi_{=}\in X} g(\Pi_{=},q_{\mathrm{Haar}}).
\end{equation}
Combining \eqref{eq:sion_minimax} and \eqref{eq:haar_least_favourable}, the adversarial-unknown model has the same value as
the Haar-random model, i.e. without loss of generality we can Haar-twirl the inputs in $\QSI_\mu^{p}$.

The same minimax-and-symmetry argument applies to $\widetilde{\QSI}_\mu^{p}$: the objective is affine in the distribution over permutations, and averaging over $S_n$ allows one to assume the unknown permutation is uniformly random.

\subsubsection{\texorpdfstring{Warm-up: optimality for $\mu = (n-h,h)$ with $p=1/2$}{Warm-up: optimality for mu = (n-h,h) with p=1/2}}

Next, consider the most common assumptions $\mu = (n-h,h)$ with $p=1/2$. The $\QSI^p_n$ problem has then the following solution:

\begin{theorem} \label{thm:permtest_optimal_1/2_binary}
    Given $n$ systems and a promise $\mu = (n-h,h)$ with $n/2 \geq h \geq 1$, $p=1/2$, the permutation test, defined by $\Pi_{=} \defeq \Pi_{(n)}$, is optimal and has an average success probability equal 
    \begin{equation}
        \mathbb{P}_{succ}(n,h) = 1 - \frac{1}{2 \binom{n}{h}}.
    \end{equation}
    One-sided success probabilities (completeness and soundness respectively) are 
    \begin{equation}
        \mathbb{P}_{c}(n,h) = 1, \qquad \mathbb{P}_{s}(n,h) = 1 - \frac{1}{\binom{n}{h}}.
    \end{equation}
\end{theorem}
\begin{proof}
Note that the permutation test is always correct on equal inputs, meaning that the completeness probability is
\begin{equation}
    \mathbb{P}_{c}(n,h) = \Tr[\Pi_= \rho_{=}^{n}] =  \Tr[\Pi_{(n)} \rho_{=}^{n}] = 1.
\end{equation}
It is only sometimes wrong on orthogonal inputs, the soundness probability is
\begin{equation}
    \mathbb{P}_{s}(n,h) = \Tr[\Pi_{\neq} \rho_{\neq}^{n,h}] = 1 - \Tr[\rho_{\neq}^{n,h}\Pi_{(n)}] = 1 - \frac{1}{\binom{n}{h}},
\end{equation}
where we used \Cref{lem:main_lemma} and $\rho_{\neq}^{n,h} \defeq \rho_{\neq}^{\mu}$ for $\mu = (n-h,h)$.
Therefore the permutation test is a feasible solution to the primal SDP (\ref{def:SDP}) with the average success probability
\begin{equation}
    \mathbb{P}_{succ}(n,h) = 1 - \frac{1}{2 \binom{n}{h}},
\end{equation}which thus lower-bounds the SDP. We claim that the permutation test is the optimal solution for the SDP. To show this, it suffices to find a feasible solution to the dual problem that attains the value of the permutation test. 

\noindent As an educated guess, we will show that
\begin{align}
    Y \defeq \frac{1}{2} \rho_{\neq}^{n,h} + \frac{1}{2}\left( 1 - \frac{1}{\binom{n}{h}} \right) \rho_{=}^{n}
\end{align}
is a feasible solution to the dual SDP (\ref{def:SDP}). Indeed $Y$ is Hermitian, and trivially
\begin{align}
    Y - \frac{1}{2}\rho_{\neq}^{n,h} = \frac{1}{2}\left( 1 - \frac{1}{\binom{n}{h}} \right) \rho_{=}^{n} \succeq 0
\end{align}
Now we need to show that
\begin{align} \label{equalsdpcondition}
    Y - \frac{1}{2}\rho_{=}^{n} = \frac{1}{2} \rho_{\neq}^{n,h} - \frac{1}{2\binom{n}{h}} \rho_{=}^{n} \succeq 0.
\end{align}

\noindent We can use \Cref{lem:main_lemma} for the case $\mu = (n-h,h)$ with $n/2 \geq h$ to deduce that
\begin{align}
    \rho_{\neq}^{n,h} &= \frac{1}{\binom{n}{h}} \sum_{\substack{\lambda \pt_d n}} \frac{d_\lambda}{m_{\lambda,d}} \of*{\Pi_{(n-h)} \otimes \Pi_{(h)}} \Pi_\lambda \\
    &= \frac{1}{\binom{n}{h}} \frac{1}{\binom{n+d-1}{n}} \Pi_{(n)} + \frac{1}{\binom{n}{h}} \sum_{\substack{\lambda \vdash_d n \\ \lambda \neq (n)}} \frac{d_\lambda}{m_{\lambda,d}} \of*{\Pi_{(n-h)} \otimes \Pi_{(h)}} \Pi_\lambda,
\end{align}
where we used the fact $\of*{\Pi_{(n-h)} \otimes \Pi_{(h)}} \Pi_{(n)} = \Pi_{(n)}$ and $m_{(n),d} = \binom{n+d-1}{n}$. Therefore,
\begin{equation}
    Y - \frac{1}{2}\rho_{=}^{n} = \frac{1}{2} \rho_{\neq}^{n,h} - \frac{1}{2\binom{n}{h}} \rho_{=}^{n} = \frac{1}{2\binom{n}{h}} \sum_{\substack{\lambda \vdash_d n \\ \lambda \neq (n)}} \frac{d_\lambda}{m_{\lambda,d}} \of*{\Pi_{(n-h)} \otimes \Pi_{(h)}} \Pi_\lambda \succeq 0,
\end{equation}
since all $\of*{\Pi_{(n-h)} \otimes \Pi_{(h)}} \Pi_\lambda$ are mutually orthogonal projectors (the isotypic projectors $\Pi_\lambda$ commute with $\Pi_{(n-h)} \otimes \Pi_{(h)}$, so the product of these two projectors is again a projector, and distinct $\lambda$ give orthogonal supports). 
Therefore, all the constraints of the dual SDP (\ref{def:SDP}) are satisfied and the objective value is the same as for the primal SDP. 
Therefore from weak SDP duality, it follows that the permutation test is optimal. 
\end{proof}

The above proof shows an optimality of the permutation test in a relaxed \emph{two-sided error} setting, thus solving an open problem from~\cite{kada2008efficiency}. 

\subsubsection{The general setting}
Now we consider a general situation: an arbitrary $p \in [0,1]$ and an arbitrary $\mu \pt n$.

\begin{theorem} \label{thm:permtest_trivial_optimal}
   For $\hyperref[def:QSI]{\QSI_\mu^{p}}$, given $n$ systems and a partition $\mu \pt n$, we define $p^*(\mu) \defeq \frac{1}{1+\binom{n}{\mu}}$. Then the primal SDP (\ref{def:SDP}) has the following optimal solution:
   \begin{equation} \label{def:opt_pvm_general}
       \Pi_= = 
       \begin{cases}
           \Pi_{(n)} &\text{ if } p^*(\mu) \leq p \leq 1, \\
           0 &\text{ if } 0 \leq p \leq p^*(\mu), 
       \end{cases}
   \end{equation}
   with average success probability
    \begin{equation} \label{def:opt_value}
        \mathbb{P}_{succ}(\mu,p) = \begin{cases}
        1 - \frac{1-p}{\binom{n}{\mu}}  &\text{ if } p^*(\mu) \leq p \leq 1, \\
        1 - p &\text{ if } 0 \leq p \leq p^*(\mu). 
        \end{cases}
    \end{equation}
    One-sided success probabilities (completeness and soundness respectively) are 
    \begin{align}
        \mathbb{P}_{c}(\mu,p) &= \begin{cases}
        1 & \text{ if } p^*(\mu) \leq p \leq 1, \\
        0 & \text{ if } 0 \leq p \leq p^*(\mu), 
        \end{cases} \qquad
        \mathbb{P}_{s}(\mu,p) = \begin{cases}
        1 - \frac{1}{\binom{n}{\mu}}& \text{ if } p^*(\mu) \leq p \leq 1, \\
        1& \text{ if } 0 \leq p \leq p^*(\mu).
        \end{cases}
    \end{align}
\end{theorem}
\begin{proof}
    The idea for the proof is the same as for \Cref{thm:permtest_optimal_1/2_binary}: we are going to present a feasible solution for the dual, which matches the feasible solution for the primal SDP, defined by \cref{def:opt_pvm_general}.
    Consider 
    \begin{equation} \label{def:Y_guess_general}
        Y \defeq \begin{cases}
        (1-p) \rho^{\mu}_{\neq} + \of*{p - \frac{1-p}{\binom{n}{\mu}}}\rho^{n}_{=} & \text{ if } p^*(\mu) \leq p \leq 1, \\
        (1-p) \rho^{\mu}_{\neq} & \text{ if } 0 \leq p \leq p^*(\mu). 
        \end{cases}
    \end{equation}
    We are going to show that this $Y$ is a feasible solution for the dual SDP. Consider the two cases:
    \begin{itemize}
        \item $0 \leq p \leq p^*(\mu)$. Then from the definition (\ref{def:Y_guess_general}) $Y - (1-p) \rho^{\mu}_{\neq} \succeq 0$ is trivially satisfied. Checking another constraint is also easy using \Cref{lem:main_lemma}.
        \begin{align}
            Y - p \rho^{n}_{=} &= \of*{\frac{1-p}{\binom{n}{\mu}} - p} \frac{\Pi_{(n)}}{\binom{n+d-1}{n}}  + \frac{1-p}{\binom{n}{\mu}} \sum_{\substack{\lambda \vdash_d n \\ \lambda \neq (n)}} \frac{d_\lambda}{m_{\lambda,d}} \of*{\Pi_{(\mu_1)} \otimes \dotsc \otimes \Pi_{(\mu_d)}} \Pi_\lambda,
        \end{align}
        but it is easy to see that 
        \begin{equation}
            \frac{1-p}{\binom{n}{\mu}} - p = \frac{1}{\binom{n}{\mu}} \of*{1 - \frac{p}{p^*(\mu)}} \geq 0,
        \end{equation}
        therefore $Y - p \rho^{n}_{=} \succeq 0$.
        \item $p^*(\mu) \leq p \leq 1$. Again, it is obvious from the definition (\ref{def:Y_guess_general}) that 
        \begin{equation}
            Y - (1-p) \rho^{\mu}_{\neq} =\frac{1}{\binom{n}{\mu}} \of*{\frac{p}{p^*(\mu)} - 1} \rho^{n}_{=} \succeq 0.
        \end{equation}
        Another constraint is similar to the previous case:
        \begin{align}
            Y - p \rho^{n}_{=} &= (1-p) \rho^{\mu}_{\neq} - \frac{1-p}{\binom{n}{\mu}} \rho^{n}_{=} 
            \\
            &= \frac{1-p}{\binom{n}{\mu}} \sum_{\substack{\lambda \vdash_d n \\ \lambda \neq (n)}} \frac{d_\lambda}{m_{\lambda,d}} \of*{\Pi_{(\mu_1)} \otimes \dotsc \otimes \Pi_{(\mu_d)}} \Pi_\lambda \succeq 0
        \end{align}
    \end{itemize}
    It remains to verify that $\Tr[Y]$ matches the primal objective value \eqref{def:opt_value}. For $p^*(\mu) \leq p \leq 1$: $\Tr[Y] = (1-p) + p - \frac{1-p}{\binom{n}{\mu}} = 1 - \frac{1-p}{\binom{n}{\mu}}$. For $0 \leq p \leq p^*(\mu)$: $\Tr[Y] = (1-p)$. Both match \cref{def:opt_value}, so weak SDP duality proves the optimality of the solution \cref{def:opt_pvm_general}.
\end{proof}

As mentioned before, the permutation test does not use any information about the order of the input state. More precisely, since $\Pi_{=} = \Pi_{(n)}$ commutes with all permutations $\psi(\pi)$, we have $\Tr[\Pi_{(n)} \tilde{\rho}_{\neq}^{\mu}] = \Tr[\Pi_{(n)} \rho_{\neq}^{\mu}]$, so the permutation test achieves the same success probability on the permuted input $\tilde{\rho}_{\neq}^{\mu}$ as on the ordered input ${\rho}_{\neq}^{\mu}$. Thus, surprisingly, \textit{knowing} the order of the unequal inputs does not allow one to achieve a higher success probability. Since the problem $\hyperref[def:QSI]{\widetilde{\QSI}_\mu^{p}}$ with randomly permuted inputs is a \textit{strictly} harder problem (any test for $\widetilde{\QSI}_\mu^{p}$ is also valid for $\QSI_\mu^{p}$, so the optimal value for $\widetilde{\QSI}_\mu^{p}$ is at most that for $\QSI_\mu^{p}$), we immediately get the following corollary.

\begin{corollary}
    For $\hyperref[def:QSI]{\widetilde{\QSI}_\mu^{p}}$, given $n$ systems and a permuted partition $\tilde{\mu} \pt n$, we define $p^*(\tilde{\mu}) \defeq \frac{1}{1+\binom{n}{\tilde{\mu}}}$. Then the same measurement as in Theorem \ref{thm:permtest_trivial_optimal} is optimal, with the same success probabilities.
\end{corollary}

\subsubsection{Hypothesis testing interpretation}

\begin{figure}[ht]
    \centering
    \begin{tikzpicture}
    \begin{axis}[
        width=0.78\textwidth,
        height=0.65\textwidth,
        xlabel={$\alpha = 1 - \mathbb{P}_{c}$},
        ylabel={$\beta = 1 - \mathbb{P}_{s}$},
        xmin=0, xmax=1.08,
        ymin=0, ymax=1.08,
        xtick={0,1.0},
        xticklabels={$0$,$1$},
        ytick={0,0.25,0.75,1.0},
        yticklabels={$0$,$\tfrac{1}{\binom{n}{\mu}}$,$1\!-\!\tfrac{1}{\binom{n}{\mu}}$,$1$},
        axis lines=left,
        every axis plot/.append style={thick},
        clip=false,
    ]

    \addplot[fill=blue!8, draw=none, forget plot]
        coordinates {(0,0.25) (1,0) (1,0.75) (0,1) (0,0.25)};

    \addplot[color=blue, thick, domain=0:1, samples=2, forget plot]
        {(1-x)/4};

    \addplot[color=blue!50, thin, dashed, domain=0:1, samples=2, forget plot]
        {1-x/4};

    \addplot[color=blue!30, thin, dashed, forget plot]
        coordinates {(0,0.25) (0,1)};

    \addplot[color=blue!30, thin, dashed, forget plot]
        coordinates {(1,0) (1,0.75)};

    \addplot[color=red, dashed, thick, domain=0:0.375, samples=2, forget plot]
        {0.25 - (2/3)*x};

    \addplot[color=red, dashdotted, thick, domain=0.40:1, samples=2, forget plot]
        {(1/9)*(1-x)};

    \addplot[only marks, mark=*, mark size=3pt, color=blue, forget plot]
        coordinates {(0, 0.25)};
    \addplot[only marks, mark=*, mark size=3pt, color=black, forget plot]
        coordinates {(1, 0)};
    \addplot[only marks, mark=*, mark size=2.5pt, color=black!50, forget plot]
        coordinates {(0, 1)};
    \addplot[only marks, mark=*, mark size=2.5pt, color=black!50, forget plot]
        coordinates {(1, 0.75)};

    \node[anchor=west, font=\small, color=blue] at (axis cs:0.00,0.28)
        {Perm.\ test};
    \node[anchor=north, font=\small] at (axis cs:1.1,0.08)
        {Trivial strategy};
    \node[anchor=west, font=\small] at (axis cs:1.03,0.75)
        {$(1,\,1\!-\!\tfrac{1}{\binom{n}{\mu}})$};

    \node[font=\small, color=blue!50!black] at (axis cs:0.4,0.55) {$\mathcal{R}(\mu)$};

    \node[font=\small, color=blue, anchor=north, rotate=-10] at (axis cs:0.55,0.22)
        {$\beta = \tfrac{1-\alpha}{\binom{n}{\mu}}$};

    \node[font=\small, color=blue!50, anchor=south, rotate=-10] at (axis cs:0.45,0.90)
        {$\beta = 1-\tfrac{\alpha}{\binom{n}{\mu}}$};

    \node[font=\small, color=red, anchor=south west, rotate=-23] at (axis cs:0.12,0.08)
        {$p > p^*$};
    \node[font=\small, color=red, anchor=south west,  rotate=-03] at (axis cs:0.40,0.06)
        {$p < p^*$};

    \end{axis}
    \end{tikzpicture}
    \caption{The hypothesis testing region $\mathcal{R}(\mu)$ (shaded parallelogram) for $n=4$ and $\mu=(3,1)$. 
    The solid blue line is the optimal boundary $\beta = (1-\alpha)/\binom{n}{\mu}$ from \Cref{prop:ht_region}; the dashed blue line is the upper boundary $\beta = 1 - \alpha/\binom{n}{\mu}$. 
    The red lines are level sets of the function $f(\alpha,\beta) = p\alpha + (1-p)\beta$ for different values of the parameter $p$: dashed line is minimizer for $p > p^*(\mu)$ (optimum at permutation test), dash-dotted is the minimizer for $p < p^*(\mu)$ (optimum at trivial test).
    }
    \label{fig:ht_region}
\end{figure}

The results of \Cref{thm:permtest_trivial_optimal} can be reformulated in the language of quantum hypothesis testing. Given the two hypotheses $H_0: \rho_{=}^{n}$ (equal) and $H_1: \rho_{\neq}^{\mu}$ (unequal) and a POVM $\{\Pi_{=}, \id - \Pi_{=}\}$, the completeness and soundness probabilities from \Cref{thm:permtest_trivial_optimal} are related to the type~I and type~II errors by
\begin{equation}
    \alpha \defeq 1 - \mathbb{P}_{c} = 1 - \Tr[\Pi_{=} \rho_{=}^{n}], \qquad \beta \defeq 1 - \mathbb{P}_{s} = \Tr[\Pi_{=} \rho_{\neq}^{\mu}].
\end{equation}
The \emph{hypothesis testing region} $\mathcal{R}(\mu) \subseteq [0,1]^2$ is the set of all achievable pairs $(\alpha,\beta)$.

\textit{Asymmetric setting.} In the \emph{one-sided error} (or asymmetric) setting, one fixes one error type and optimises the other. For instance, requiring perfect completeness ($\alpha = 0$) and minimising $\beta$ recovers the known result of~\cite{kada2008efficiency}: the permutation test is optimal with $\beta = 1/\binom{n}{\mu}$, i.e.\ $\mathbb{P}_{s} = 1-1/\binom{n}{\mu}$.

\textit{Symmetric setting.} In the \emph{Bayesian} (or symmetric) setting with prior $\{p, 1-p\}$, one maximises the average success probability $\mathbb{P}_{succ} = p\,\mathbb{P}_{c} + (1-p)\,\mathbb{P}_{s} = 1 - p\alpha - (1-p)\beta$. \Cref{thm:permtest_trivial_optimal} fully solves this: the permutation test ($\alpha=0$, $\beta = 1/\binom{n}{\mu}$) is optimal for $p \geq p^*(\mu)$, and the trivial test ($\alpha = 1$, $\beta = 0$) is optimal for $p \leq p^*(\mu)$.

These results determine the complete trade-off between $\alpha$ and $\beta$ in the full achievable region.

\begin{proposition}\label{prop:ht_region}
    The achievable region $\mathcal{R}(\mu)$ is the parallelogram
    \begin{equation}\label{eq:ht_region}
        \mathcal{R}(\mu) = \set*{(\alpha,\beta) \in [0,1]^2 \;\Big|\; \frac{1-\alpha}{\binom{n}{\mu}} \leq \beta \leq 1 - \frac{\alpha}{\binom{n}{\mu}}}.
    \end{equation}
    In particular, the lower boundary (optimal trade-off) is
    \begin{equation}\label{eq:ht_boundary}
        \beta = \frac{1 - \alpha}{\binom{n}{\mu}}, \qquad \alpha \in [0,1],
    \end{equation}
    and the upper boundary is $\beta = 1 - \alpha/\binom{n}{\mu}$. Equivalently, the optimal completeness--soundness trade-off is $\mathbb{P}_{s} = 1 - \mathbb{P}_{c}/\binom{n}{\mu}$.
\end{proposition}
\begin{proof}
    Define the operator $M \defeq \rho_{\neq}^{\mu} - \frac{1}{\binom{n}{\mu}}\rho_{=}^{n}$. The proof of \Cref{thm:permtest_trivial_optimal} establishes that $M \succeq 0$ (cf.\ the dual feasibility condition at $p = p^*(\mu)$), and $\Tr[M] = 1 - 1/\binom{n}{\mu}$. For any POVM element $0 \preceq \Pi_{=} \preceq \id$, we can write
    \begin{equation}
        \beta = \Tr[\Pi_{=}\,\rho_{\neq}^{\mu}] = \frac{1}{\binom{n}{\mu}}\Tr[\Pi_{=}\,\rho_{=}^{n}] + \Tr[\Pi_{=}\, M] = \frac{1-\alpha}{\binom{n}{\mu}} + \Tr[\Pi_{=}\, M].
    \end{equation}

    \textit{Lower boundary.} Since $\Pi_{=} \succeq 0$ and $M \succeq 0$, we have $\Tr[\Pi_{=}\, M] \geq 0$, hence $\beta \geq (1-\alpha)/\binom{n}{\mu}$. Equality holds for $\Pi_{=} = (1-\alpha)\,\Pi_{(n)}$, since $\Pi_{(n)} M \Pi_{(n)} = 0$ (because $\Pi_{(n)}\rho_{\neq}^{\mu}\Pi_{(n)} = \rho_{=}^{n}/\binom{n}{\mu}$, see \Cref{lem:main_lemma}).

    \textit{Upper boundary.} Since $\Pi_{=} \preceq \id$ and $M \succeq 0$, we have $\Tr[\Pi_{=}\, M] \leq \Tr[M] = 1 - 1/\binom{n}{\mu}$, hence $\beta \leq 1 - \alpha/\binom{n}{\mu}$. Equality holds for $\Pi_{=} = (1-\alpha)\,\Pi_{(n)} + (\id - \Pi_{(n)})$, which satisfies $0 \preceq \Pi_= \preceq \id$ and gives
    \begin{align}
        \Tr[\Pi_{=}\,\rho_{=}^{n}] &= 1-\alpha, \\
        \Tr[\Pi_{=}\,\rho_{\neq}^{\mu}] &= \frac{1-\alpha}{\binom{n}{\mu}} + \Tr[\rho_{\neq}^{\mu}] - \Tr[\Pi_{(n)}\rho_{\neq}^{\mu}] = \frac{1-\alpha}{\binom{n}{\mu}} + 1 - \frac{1}{\binom{n}{\mu}} = 1 - \frac{\alpha}{\binom{n}{\mu}}.
    \end{align}

    \textit{Achievability.} For any $\alpha \in [0,1]$ and $\beta_0 \in \bigl[\frac{1-\alpha}{\binom{n}{\mu}},\; 1-\frac{\alpha}{\binom{n}{\mu}}\bigr]$, choose $\Pi_{=} = (1-\alpha)\,\Pi_{(n)} + t\,(\id - \Pi_{(n)})$ with $t = \bigl(\beta_0 - \frac{1-\alpha}{\binom{n}{\mu}}\bigr)\big/\bigl(1-\frac{1}{\binom{n}{\mu}}\bigr) \in [0,1]$. This satisfies $0 \preceq \Pi_{=} \preceq \id$ and achieves $(\alpha,\beta_0)$.
\end{proof}

\noindent The region $\mathcal{R}(\mu)$ is depicted in \Cref{fig:ht_region}. 
It is a parallelogram with four vertices: $(0, 1/\binom{n}{\mu})$ (permutation test), $(1,0)$ (trivial test), $(1,1-1/\binom{n}{\mu})$ (the ``anti-test'' $\Pi_{=} = \id - \Pi_{(n)}$), and $(0,1)$ (always output ``equal''). Both boundaries have slope $-1/\binom{n}{\mu}$. In the symmetric setting, the level sets of $\mathbb{P}_{succ} = 1 - p\alpha - (1-p)\beta$ are lines of slope $-p/(1-p)$. The threshold prior $p^*(\mu) = 1/(1+\binom{n}{\mu})$ is precisely the value at which $p/(1-p) = 1/\binom{n}{\mu}$, so the level set is parallel to both boundaries. For $p > p^*(\mu)$ the level sets are steeper, so the optimum is at the permutation test vertex; for $p < p^*(\mu)$ they are flatter, so the optimum is at the trivial test vertex --- recovering \Cref{thm:permtest_trivial_optimal}.

\begin{remark}[Comparison with parallel repetition of the Swap test]
A common use of the Swap test is to compare the outputs of two procedures by applying it independently to multiple pairs of states and amplifying the success probability via parallel repetition. In our setting, this corresponds to the following task: given $n$ copies of $\ket{\psi}$ and $n$ copies of $\ket{\phi}$, decide whether $\ket{\psi} = \ket{\phi}$ or $\langle \psi | \phi \rangle = 0$. This is precisely the instance $\mu = (n,n)$ of the quantum state identity problem on $2n$ systems with prior $p = \frac{1}{2}$.

Repeating the Swap test on $n$ pairs yields perfect completeness and soundness $1 - 2^{-n}$, so the average success probability is
\begin{align}
P_{\mathrm{swap}}(n) = 1 - \frac{1}{2^{n+1}}.
\end{align}
In contrast, the permutation test achieves the optimal success probability
\begin{align}
P_{\mathrm{perm}}(n) = 1 - \frac{1}{2 \binom{2n}{n}}.
\end{align}
Using the standard estimate $\binom{2n}{n} \sim \frac{4^n}{\sqrt{\pi n}}$, we obtain
\begin{align}
\frac{1}{2 \binom{2n}{n}} \sim \frac{\sqrt{\pi n}}{2 \cdot 4^n}.
\end{align}

To reach error at most $\varepsilon$, the Swap test therefore requires
\begin{align}
2^{-(n+1)} \le \varepsilon 
&\Longleftrightarrow 
n \gtrsim \log_2(1/\varepsilon),
\end{align}
while the permutation test requires
\begin{align}
\frac{\sqrt{\pi n}}{2 \cdot 4^n} \le \varepsilon 
&\Longleftrightarrow 
4^n \gtrsim \frac{1}{\varepsilon} 
\Longleftrightarrow 
n \gtrsim \tfrac{1}{2} \log_2(1/\varepsilon),
\end{align}
up to lower-order terms. Hence, the permutation test achieves the same error using asymptotically a factor $2$ fewer copies of each state than parallel repetition of the Swap test.

In particular, our optimality result implies that this improvement is information-theoretically tight: no protocol can achieve a higher success probability with fewer copies.
\end{remark}

\subsection{\texorpdfstring{Performance of the $G$-test}{Performance of the G-test}}
\label{sec:G_test_perf}
We can generalize the idea of the permutation test to an arbitrary subgroup $G \subseteq \S_n$. Assuming that we want to distinguish between $\tilde{\rho}_{\neq}^{\mu}$ and $\rho_{=}^{n}$, we call the \emph{$G$-test} the measurement given by
\begin{equation}
    \Pi_= = \Pi_G, \qquad \Pi_{\neq} = \id - \Pi_G,
\end{equation}
where $\Pi_G$ is the projector onto the trivial irrep of $G$:
\begin{equation}
    \Pi_G \defeq \frac{1}{\abs{G}} \sum_{\pi \in G} \psi(\pi).
\end{equation}
For $G = \S_n$ we recover the usual permutation test $\Pi_G = \Pi_{(n)}$. The $G$-test has the perfect completeness probability:
\begin{equation}
    \mathbb{P}^{G}_{c}(\mu) = \Tr[\rho_{=}^{n} \Pi_G] = \frac{ \Tr[\Pi_{(n)} \Pi_G ]}{\binom{n+d-1}{n}} = \frac{\Tr[\Pi_{(n)}]}{\binom{n+d-1}{n}} = 1,
\end{equation}
because for every $\pi \in G$ we have $\pi\Pi_{(n)}=\Pi_{(n)}$.
Due to \Cref{lem:main_lemma} we have
\begin{equation}
    \Tr[\tilde{\rho}_{\neq}^{\mu}\Pi_G] = \frac{1}{\binom{n}{\mu}}\sum_{\substack{\lambda \pt_d n}} \frac{K_{\lambda,\mu} \Tr[\Pi_\lambda \Pi_G]}{m_{\lambda,d}},
\end{equation}
so the soundness probability of the $G$-test can be calculated to be
\begin{equation}\label{eq:G_test_formula}
    \mathbb{P}^{G}_{s}(\mu) = 1 - \frac{1}{\binom{n}{\mu}} \sum_{\substack{\lambda \pt_d n}} K_{\lambda,\mu} r^{G}_{\lambda},
\end{equation}
where $r^{G}_{\lambda}$ is the multiplicity of the trivial irrep of the subgroup $G \subseteq \S_n$ inside the irrep $\lambda$ of the symmetric group $\S_n$. 

\subsection{The Circle test}

In this section, we consider the circle test, i.e. the case when $G$ is the cyclic group $\textup{C}_n \defeq \ipp{(12\cdots n)} \subset \S_n$. Then the numbers $r^{\textup{C}_n}_{\lambda}$ can be understood combinatorially \cite{kraskiewicz2001algebra}:
\begin{equation}
    r^{\textup{C}_n}_{\lambda} = \abs{\set{T \in \SYT(\lambda) \, | \, \maj_n(T) = 0}},
\end{equation}
where $\maj_n(T)$ is the \emph{major index modulo $n$} of the \emph{standard Young tableaux} $T$. Moreover, one can obtain the following identity \cite[Theorem 4]{swanson2018existence} for $\lambda \pt n$:
\begin{equation}
    \frac{r^{\textup{C}_n}_{\lambda}}{d_\lambda} = \frac{1}{n} + \frac{1}{n} \sum_{\substack{\ell | n \\ \ell \neq 1}} \frac{\chi_{\lambda}(\ell^{n/\ell})}{d_\lambda} \phi(\ell),
\end{equation}
where $\phi(\ell)$ is the Euler's totient function and $\chi_{\lambda}(\ell^{n/\ell})$ is the symmetric group $\S_n$ character for the conjugacy class of type $(\ell,\dotsc,\ell)$. Now since $\sum_{\substack{\lambda \pt_d n}} K_{\lambda,\mu} d_{\lambda} = \binom{n}{\mu}$ we get the general soundness error probability
\begin{equation} \label{sound:circle}
    1 - \mathbb{P}^{G}_{s}(\mu) = \frac{1}{n} + \frac{1}{n}  \sum_{\substack{\lambda \pt_d n}} \frac{K_{\lambda,\mu} d_\lambda}{\binom{n}{\mu}} \sum_{\substack{\ell | n \\ \ell \neq 1}} \frac{\chi_{\lambda}(\ell^{n/\ell})}{d_\lambda} \phi(\ell).
\end{equation}

It is easy to see that $r^{\textup{C}_n}_{(n)} = 1$ and $r^{\textup{C}_n}_{(n-1,1)} = 0$. Therefore, for $\mathbb{P}^{G}_{s}(n,h) \defeq \mathbb{P}^{G}_{s}((n-h,h))$ we have from \cref{eq:G_test_formula}:
\begin{align}
    \mathbb{P}^{G}_{s}(n,1) = 1 - \frac{1}{n}.
\end{align}
The worst case analysis of the circle test for prime $n$ and for asymptotic $n$ was performed in \cite{kada2008efficiency}. \Cref{sound:circle} suggests that more fine-grained analysis is possible via Roichmann's bound~\cite{roichman1996upper} and other results of \cite{swanson2018existence}.

\subsection{The Iterated Wreath test} 
\label{sec:def_wreath_test}
We now initiate the study of the $G$-test when $G$ is an iterated wreath product of symmetric groups. More precisely, if we have an arbitrary number of systems $n$ with the prime factorization $n = \prod_{i=1}^{r}p_i^{m_i}$ then we define $G$ as an iterated wreath product: \begin{equation}
    G \defeq \of{\S_{p_1} \wr^{m_1} \S_{p_1}} \wr \of{\S_{p_2} \wr^{m_2} \S_{p_2}} \wr \dotsc \wr \of{\S_{p_r} \wr^{m_r} \S_{p_r}},
\end{equation}
where we used the notation
\begin{equation}
    \S_{p} \wr^{m} \S_{p} \defeq \underbrace{\S_{p} \wr \dotsc \wr \S_{p}}_{m}.
\end{equation}

Note, that, for this group, two values of $r^G_\lambda$, needed for \cref{eq:G_test_formula}, can be determined immediately. First, $r^G_{(n)} = 1$ trivially, since the trivial representation of $\S_n$ restricts to the trivial representation of any subgroup. 
Second, $G$ acts transitively on $[n]$: each factor $\S_{p_i} \wr^{m_i} \S_{p_i}$ is the automorphism group of a complete $p_i$-ary rooted tree acting on its $p_i^{m_i}$ leaves, and wreath products of transitive groups are transitive. 
By Burnside's lemma, the average number of fixed points of elements of $G$ on $[n]$ therefore equals $1$. 
Since the character of the standard representation satisfies $\chi_{(n-1,1)}(g) = \mathrm{fix}(g) - 1$, where $\mathrm{fix}(g)$ is the number of fixed points of $[n]$ for $g$, we obtain
\begin{equation}
    r^G_{(n-1,1)} = \frac{1}{|G|}\sum_{g\in G}\bigl(\mathrm{fix}(g)-1\bigr) = 1 - 1 = 0.
\end{equation}
Substituting into \cref{eq:G_test_formula} with $\mu = (n-1,1)$ then gives $\mathbb{P}^G_s(n,1) = 1 - 1/n$, so the iterated wreath $G$-test is already optimal for $h=1$, matching both the circle test and the permutation test.

For $h \geq 2$, computing $r^G_\lambda$ is significantly harder. 
While the irreducible representations of iterated wreath products are well studied — they are classified by labeled rooted trees \cite{orellana2004rooted, im2018generalized}, with branching rules described via Bratteli diagrams — these results concern the internal representation theory of $G$ as an abstract group. 
The quantity $r^G_\lambda$ requires instead the restriction of an $\S_n$-irrep $\lambda$ to $G \subset \S_n$, which is not captured by these results and for which no closed-form formula is known in the literature. 
This motivates the approach of the next section, where we study the special case $n = 2^m$, so $G = \S_2 \wr^m \S_2$. 
We refer to this case as \emph{Iterated Swap Tree}. 
Rather than computing $r^G_\lambda$ exactly, we develop a recursive argument that bounds the protocol's soundness.

\section{Approximation of the Permutation test by the Iterated Swap Tree}
\subsection{The Iterated Swap Tree protocol}
We propose a new simple protocol to compute the $\hyperref[def:QSI]{\widetilde{\QSI}_\mu^{p}}$ problem based only on applying the Swap test in a tree-like way on all the inputs. One motivation is that this protocol will always be correct on equal inputs and hence has perfect completeness. 
Another motivation is practical, applying a single Swap test requires only controlling a single gate while the circle test requires control of $n$ shifts. 
The permutation test even requires controlling over $n!$ permutations. 
Not only is the controlled operation easier, but also the Fourier transform one has to apply afterwards to read out the correct measurement is only on a single qubit, instead of $\log_2(n)$ qubits for the circle test or $\log_2(n!)$ qubits for the permutation test. 
Unfortunately, we do not know of any way to compute the quantities in the aforementioned $G$-test to get the exact performance of the protocol. However, we give an alternative way to still bound its performance and show that it is similar to that of the circle test. 
Note that this operation uses $n-1$ Swap tests, answering an open question in \cite{kada2008efficiency} that asks for an efficient approximation of the permutation test using only $O(n)$ Swap tests.

\begin{custalgo}[The Iterated Swap Tree (IST)] \label{alg:IST}
\noindent \textbf{Input}: $n=2^m$ quantum states $\ket{\psi_1} \otimes \ket{\psi_2} \otimes \dots \otimes \ket{\psi_n}$, 
\begin{enumerate}
    \item Randomly permute the index labels of all $n$ input registers with a permutation $P$ to obtain the ordered set $I = \{i_1,i_2,\dots,i_n\}$ with $i_k \in [n]$.
    \item For $l \in \{0,\dots,m-1\}$:
    \begin{enumerate}
        \item Let $S_j$, where $\bigcup_{j=1}^{n/2^l} S_j = I $ be the sets of $2^l$ adjacent indices in $I$.
        \item For $b \in [n/2^{l+1}]$
        \begin{enumerate}
            \item Perform the Swap test with as inputs registers $S_{2b-1}$ and $S_{2b}$.
            \item If the measurement outcome is $1$, return `not equal', else continue.
        \end{enumerate}
    \end{enumerate}
    \item Return `equal'.
\end{enumerate}
\end{custalgo}

\tikzset{every picture/.style={line width=0.75pt}} 
\begin{figure}
    \centering

\begin{tikzpicture}[x=0.75pt,y=0.75pt,yscale=-1,xscale=1]

\draw    (53,28.33) -- (157.5,28.33) ;
\draw    (53,62.13) -- (158,62.13) ;
\draw    (153.83,62.13) .. controls (192.83,62.13) and (151.83,47.33) .. (191.83,47.33) ;
\draw    (152.83,28.33) .. controls (190.83,27.83) and (154.83,45) .. (191.83,45) ;
\draw  [fill={rgb, 255:red, 255; green, 255; blue, 255 }  ,fill opacity=1 ] (121.42,21.5) -- (159.5,21.5) -- (159.5,74.5) -- (121.42,74.5) -- cycle ;
\draw    (53,91) -- (155.5,91) ;
\draw    (53,124.79) -- (155,124.79) ;
\draw    (153.83,124.79) .. controls (192.83,124.79) and (151.83,110) .. (191.83,110) ;
\draw    (152.83,91) .. controls (190.83,90.5) and (154.83,107.67) .. (191.83,107.67) ;
\draw  [fill={rgb, 255:red, 255; green, 255; blue, 255 }  ,fill opacity=1 ] (121.42,84.17) -- (159.5,84.17) -- (159.5,137.17) -- (121.42,137.17) -- cycle ;
\draw    (224.17,47.66) .. controls (262.17,47.17) and (224.83,76.67) .. (261.83,76.67) ;
\draw    (223.5,108.13) .. controls (262.5,108.13) and (222.83,78.67) .. (262.83,78.67) ;
\draw    (225.83,110.46) .. controls (264.83,110.46) and (225.17,81) .. (265.17,81) ;
\draw    (225.42,45.48) .. controls (265.42,44.99) and (228.09,74.49) .. (265.09,74.49) ;
\draw  [fill={rgb, 255:red, 255; green, 255; blue, 255 }  ,fill opacity=1 ] (187.75,19.5) -- (225.83,19.5) -- (225.83,138) -- (187.75,138) -- cycle ;
\draw    (53,154.33) -- (154,154.33) ;
\draw    (53,188.13) -- (154,188.13) ;
\draw    (152.83,188.13) .. controls (191.83,188.13) and (150.83,173.33) .. (190.83,173.33) ;
\draw    (151.83,154.33) .. controls (189.83,153.83) and (153.83,171) .. (190.83,171) ;
\draw  [fill={rgb, 255:red, 255; green, 255; blue, 255 }  ,fill opacity=1 ] (120.42,147.5) -- (158.5,147.5) -- (158.5,200.5) -- (120.42,200.5) -- cycle ;
\draw    (53,217) -- (152.5,217) ;
\draw    (53,250.79) -- (153.5,250.79) ;
\draw    (152.83,250.79) .. controls (191.83,250.79) and (150.83,236) .. (190.83,236) ;
\draw    (151.83,217) .. controls (189.83,216.5) and (153.83,233.67) .. (190.83,233.67) ;
\draw  [fill={rgb, 255:red, 255; green, 255; blue, 255 }  ,fill opacity=1 ] (120.42,209.17) -- (158.5,209.17) -- (158.5,262.17) -- (120.42,262.17) -- cycle ;
\draw    (223.17,173.66) .. controls (261.17,173.17) and (223.83,202.67) .. (260.83,202.67) ;
\draw    (222.5,234.13) .. controls (261.5,234.13) and (221.83,204.67) .. (261.83,204.67) ;
\draw    (224.83,236.46) .. controls (263.83,236.46) and (224.17,207) .. (264.17,207) ;
\draw    (225.5,171.33) .. controls (263.5,170.83) and (226.17,200.33) .. (263.17,200.33) ;
\draw  [fill={rgb, 255:red, 255; green, 255; blue, 255 }  ,fill opacity=1 ] (186.75,145.5) -- (224.83,145.5) -- (224.83,264) -- (186.75,264) -- cycle ;
\draw    (265.09,74.49) -- (306.09,74.49) ;
\draw    (265.17,76.46) -- (306.17,76.46) ;
\draw    (265.17,78.64) -- (306.17,78.64) ;
\draw    (265.17,81) -- (306.17,81) ;
\draw    (263.24,200.49) -- (304.24,200.49) ;
\draw    (263.32,202.46) -- (304.32,202.46) ;
\draw    (263.32,204.64) -- (304.32,204.64) ;
\draw    (263.32,207) -- (304.32,207) ;
\draw  [fill={rgb, 255:red, 255; green, 255; blue, 255 }  ,fill opacity=1 ] (259.08,18.5) -- (294.17,18.5) -- (294.17,267.5) -- (259.08,267.5) -- cycle ;
\draw  [fill={rgb, 255:red, 255; green, 255; blue, 255 }  ,fill opacity=1 ] (60.08,23.5) -- (95.17,23.5) -- (95.17,266.5) -- (60.08,266.5) -- cycle ;
\draw    (405,196.33) -- (525.86,196.33) ;
\draw    (405,170.66) -- (513.67,170.66) ;
\draw    (407,231.66) -- (527.57,231.66) ;
\draw  [fill={rgb, 255:red, 255; green, 255; blue, 255 }  ,fill opacity=1 ] (436.75,184.75) -- (481,184.75) -- (481,239.67) -- (436.75,239.67) -- cycle ;
\draw  [fill={rgb, 255:red, 255; green, 255; blue, 255 }  ,fill opacity=1 ] (413.32,160.79) -- (430.57,160.79) -- (430.57,179.12) -- (413.32,179.12) -- cycle ;

\draw  [fill={rgb, 255:red, 255; green, 255; blue, 255 }  ,fill opacity=1 ] (487.75,161.12) -- (505,161.12) -- (505,179.45) -- (487.75,179.45) -- cycle ;

\draw  [fill={rgb, 255:red, 255; green, 255; blue, 255 }  ,fill opacity=1 ] (511.89,160.67) -- (529.14,160.67) -- (529.14,179.01) -- (511.89,179.01) -- cycle ;
\draw    (514.97,175.67) .. controls (517.14,170.67) and (517.7,169.95) .. (520.52,169.84) .. controls (523.34,169.73) and (523.92,171.95) .. (525.81,175.67) ;
\draw    (519.92,175.01) -- (525.03,167.34) ;
\draw  [fill={rgb, 255:red, 0; green, 0; blue, 0 }  ,fill opacity=1 ] (526.01,165.7) -- (526.08,170.2) -- (522.02,167.77) -- cycle ;

\draw    (458.39,170.67) -- (458.39,184.89) ;
\draw  [color={rgb, 255:red, 0; green, 0; blue, 0 }  ,draw opacity=1 ][fill={rgb, 255:red, 0; green, 0; blue, 0 }  ,fill opacity=1 ] (456.17,170.67) .. controls (456.17,169.44) and (457.16,168.44) .. (458.39,168.44) .. controls (459.62,168.44) and (460.61,169.44) .. (460.61,170.67) .. controls (460.61,171.89) and (459.62,172.89) .. (458.39,172.89) .. controls (457.16,172.89) and (456.17,171.89) .. (456.17,170.67) -- cycle ;
\draw    (428,56.33) -- (490.43,56.33) ;
\draw    (428,91.66) -- (491.57,91.66) ;
\draw  [fill={rgb, 255:red, 255; green, 255; blue, 255 }  ,fill opacity=1 ] (436.32,46.75) -- (480.57,46.75) -- (480.57,101.67) -- (436.32,101.67) -- cycle ;
\draw [color={rgb, 255:red, 155; green, 155; blue, 155 }  ,draw opacity=1 ] [dash pattern={on 4.5pt off 4.5pt}]  (336,25) -- (336,271) ;

\draw (450,130) node [anchor=north west][inner sep=0.75pt]    {$=$};
\draw (448,68) node [anchor=north west][inner sep=0.75pt]    {$\mathsf{ST}$};
\draw (398,83) node [anchor=north west][inner sep=0.75pt]    {$| \psi _{2} \rangle $};
\draw (398,48) node [anchor=north west][inner sep=0.75pt]    {$| \psi _{1} \rangle $};
\draw (439.45,207) node [anchor=north west][inner sep=0.75pt]  [font=\scriptsize]  {$\mathrm{SWAP}$};
\draw (383,162) node [anchor=north west][inner sep=0.75pt]    {$|0\rangle $};
\draw (378,224) node [anchor=north west][inner sep=0.75pt]    {$| \psi _{2} \rangle $};
\draw (378,188) node [anchor=north west][inner sep=0.75pt]    {$| \psi _{1} \rangle $};
\draw (70,132) node [anchor=north west][inner sep=0.75pt]    {$P$};
\draw (23,242.5) node [anchor=north west][inner sep=0.75pt]    {$| \psi _{8} \rangle $};
\draw (23,208) node [anchor=north west][inner sep=0.75pt]    {$| \psi _{7} \rangle $};
\draw (23,180) node [anchor=north west][inner sep=0.75pt]    {$| \psi _{6} \rangle $};
\draw (130,228) node [anchor=north west][inner sep=0.75pt]    {$\mathsf{ST}$};
\draw (197,198) node [anchor=north west][inner sep=0.75pt]    {$\mathsf{ST}$};
\draw (23,146) node [anchor=north west][inner sep=0.75pt]    {$| \psi _{5} \rangle $};
\draw (130,168) node [anchor=north west][inner sep=0.75pt]    {$\mathsf{ST}$};
\draw (267,132) node [anchor=north west][inner sep=0.75pt]    {$\mathsf{ST}$};
\draw (130,104) node [anchor=north west][inner sep=0.75pt]    {$\mathsf{ST}$};
\draw (197,75) node [anchor=north west][inner sep=0.75pt]    {$\mathsf{ST}$};
\draw (23,117) node [anchor=north west][inner sep=0.75pt]    {$| \psi _{4} \rangle $};
\draw (23,83) node [anchor=north west][inner sep=0.75pt]    {$| \psi _{3} \rangle $};
\draw (23,54) node [anchor=north west][inner sep=0.75pt]    {$| \psi _{2} \rangle $};
\draw (23,20) node [anchor=north west][inner sep=0.75pt]    {$| \psi _{1} \rangle $};
\draw (130,41) node [anchor=north west][inner sep=0.75pt]    {$\mathsf{ST}$};
\draw (490,164) node [anchor=north west][inner sep=0.75pt]  [font=\footnotesize]  {$H$};
\draw (415,164) node [anchor=north west][inner sep=0.75pt]  [font=\footnotesize]  {$H$};

\end{tikzpicture}
    \caption{Visualisation of the Iterated Swap Tree (\hyperref[alg:IST]{IST}) protocol for $n=8$.}
    \label{fig:SWAP-tree}
\end{figure} 

Note that Algorithm \ref{alg:IST} is always correct on the equal input case. To analyze the performance of Algorithm~\ref{alg:IST} we want to estimate the probability that the algorithm accepts the input state as being the all equal state while the input was actually a state on which a subset was orthogonal to the other states. To do so, consider what the Swap test tells us on two input states. If the states are equal the measurement outcome will \textit{always} be `0', if the states are orthogonal the measurement outcome will be either `0' or `1' with 
probability $1/2$ \cite{buhrman2001fingerprinting}. We call the cases in which the Swap test may detect that the inputs are orthogonal a \textit{possible click}. If the input states are all equal, all the Swap tests in the tree will get measurement outcome `0'. If a subset of the input states is orthogonal to the others, we are interested in the number of possible clicks, i.e. the number of times the Swap test could tell us two states are orthogonal. Suppose we apply the Swap test to two orthogonal states $\ket{\psi}, \ket{\phi}$, then the \textit{post-measurement} states will be:
\begin{align}
    \frac{1}{\sqrt{2}} \ket{\psi}\ket{\phi} &+ \ket{\phi}\ket{\psi}, &\text{or}& &\frac{1}{\sqrt{2}} \ket{\psi}\ket{\phi} &- \ket{\phi}\ket{\psi},
\end{align}
corresponding to measurement outcomes `0' and `1' respectively.
 
For example, if a single state is orthogonal to the other states, then, it is clear that in every layer of the \hyperref[alg:IST]{IST} in Figure~\ref{fig:SWAP-tree} there is a single possible click, no matter where the state is placed. The \hyperref[alg:IST]{IST} protocol will only not detect this single input being orthogonal if all the measurement outcomes were `0', which happens with probability $\left(\frac{1}{2}\right)^{\log n} = \frac{1}{n}$. If the inputs were all equal, the \hyperref[alg:IST]{IST} protocol will always be correct. Therefore if again half the time the input is all equal and the other half there is a single orthogonal state the probability of correctly identifying the state with the \hyperref[alg:IST]{IST} is:
\begin{align}
    \frac{1}{2} \cdot 1 + \frac{1}{2} \cdot \left(1 - \frac{1}{n} \right) = 1 - \frac{1}{2n}.
\end{align}
Which is equal to the $h=1$ case of Theorem \ref{thm:permtest_optimal_1/2_binary}, thus we see that for the $h=1$ case the \hyperref[alg:IST]{IST} is in fact optimal. Interestingly, the position of the orthogonal state is not used by the \hyperref[alg:IST]{IST}, as is also the case for the permutation test, but both protocols are nonetheless optimal.

Of interest is now how well the \hyperref[alg:IST]{IST} protocol performs for different input partitions. The performance depends purely on the number of possible clicks for the partition. For a specific input one can compute the number of possible clicks by checking if the input ports in every Swap test are orthogonal or not. However, it is more difficult to say something about the number of possible clicks if the inputs are randomly permuted. We want to get the soundness probability as an averaged quantity only depending on $n$ and $h$. To say something quantitatively we consider the case where there is a partition into two parts ($\mu = (n-h, h)$), and we propose a lower bound to the performance of the \hyperref[alg:IST]{IST} protocol by bounding the number of possible clicks. 

The idea is as follows, in the two input ports of a Swap test consider the number of orthogonal states in each port. If this number of states is different between the two ports, we know for sure that the input states are orthogonal to each other and we get a possible click. If the number of orthogonal states are equal it is possible to have a possible click, but as a simplification we disregard this possibility. Then, we get a lower bound on the number of possible clicks.

One can compute this lower bound on the number of possible clicks iteratively for every layer:
\begin{itemize}
    \item Check for every input pair if the number of orthogonal states in each input port is different. If yes, add one to a possible click counter. 
    \item Add the number of orthogonal states in each port and send this through to the next port.
\end{itemize}
As a simple example consider the input with $h=3,n=8$:
\begin{align}
    \ket{\phi}  \ket{\psi} \ket{\psi} \ket{\psi} \ket{\phi} \ket{\psi} \ket{\psi} \ket{\phi},
\end{align}

where $\bra{\phi} \ket{\psi} = 0$. Then we can write this input as `1,0,0,0,1,0,1,1', where `1' denotes the state $\ket{\phi}$ and `0' $\ket{\psi}$. Then there are $2$ cases where the number of orthogonal states in the ports are different. For the second layer, we add all the pairs of inputs to get `1,0,1,2', which has $2$ of these cases. Finally, for the last layer, we get `1,3', which has $1$ possible click. Thus the total number of possible clicks for this input was at least $5$. So the probability to detect that this input was the unequal input is at least $1 - \left(\frac{1}{2}\right)^5$. Counting this way gives us a recursive way to lower the possible clicks and the performance of the \hyperref[alg:IST]{IST} protocol by a specific quantity. To see why this counting strategy does not count every case, we give an example of a state initialization where the two input ports of the Swap test have the same number of orthogonal states, but nonetheless, the input states are orthogonal. Consider the input:
\begin{align}
    \ket{\phi}\ket{\phi}\ket{\psi}\ket{\psi} \ket{\phi}\ket{\psi}\ket{\phi}\ket{\psi},
\end{align}
and suppose we arrive at the final Swap test where both ports have 4 states as input. We can assume we have not had a measurement outcome `1' yet on any of the Swap tests, since we stop when we detect this. Note that both ports will have 2 orthogonal states in them, so our counting procedure would not count this case as a measurement where a possible click can occur. Then the input state in the first port is
\begin{align}
    \frac{1}{\sqrt{2}} \left( \ket{\phi\phi\psi\psi} + \ket{\psi\psi\phi\phi} \right),
\end{align}
and of the second state
\begin{align}
    \frac{1}{2}\left( \ket{\phi\psi\phi\psi} + \ket{\phi\psi\psi\phi} + \ket{\psi\phi\psi\phi} + \ket{\psi\phi\phi\psi} \right).
\end{align}
Note that both states are orthogonal to each other, so we have a possible click. So we get a lower bound on the number of possible clicks. The recursive structure of counting this way allows us to derive a general upper bound for the soundness error of the Iterated Swap Tree.

\subsection{A recursive upper bound}

\begin{theorem}
    The soundness probability of the \hyperref[alg:IST]{IST} for $n=2^m$ and $h \in [n]$, with the promise that $\mu = (n-h,h)$ is 
    \begin{equation}
        \mathbb{P}^{\mathrm{IST}}_{s}(n,h) \geq 1 - \frac{\gamma(h,\log_2(n))}{\binom{n}{h}},
    \end{equation}
    where $\gamma(h,m)$ satisfies the recurrence relation for every integer $h \geq 2$ and integer $m \geq 0$
    \begin{equation}
        \gamma(h,m) = \sum_{k=0}^{\floor{\frac{h}{2}}} \gamma(k,m-1) \gamma(h-k,m-1)
    \end{equation}
    with the boundary conditions $\gamma(0,m) = \gamma(1,m) = 1$ for every integer $m \geq 0$ and $\gamma(h,0) = 0$ for every integer $h \geq 2$.
\label{thm:swap_trees}
\end{theorem}
\begin{proof}

Note that $n/2 \geq h$ by definition, otherwise we could say the other set of states is the orthogonal part. For $n = 2^m$ and a fixed $h$, we know that the probability to \textit{not} detect will be $\mathbb{E}[\left( \frac{1}{2} \right)^{\#\text{poss. clicks}(h,m)}]$, where the expected value is taken over all possible input permutations. For $n = 2^m$ and the number of orthogonal states $h$, we define $\alpha(h,m)$ as $\frac{1}{2}$ to the power of our value of the counting process described above averaged over all permutations. Since our counting process undercounts the number of clicks, we have $\alpha(h,m) \geq \mathbb{E}[\left( \frac{1}{2} \right)^{\#\text{poss. clicks}(h,m)}]$. For a specific input we can compute the final value of our counting process recursively as follows:
\begin{itemize}
    \item Divide the input into two equally sized halves and count the number of orthogonal states in each of the halves. If the inputs are of size 1 stop.
    \item These two numbers will be exactly the number of orthogonal states in each port in the last layer of the protocol.
    \item If the numbers are different add 1 to the counter, else do nothing.
    \item Now do this process again to both of the halves of the original input.
\end{itemize}
The insight to why this recursion works is that if the final layer is removed we have two \hyperref[alg:IST]{IST} independent from each other on the first and second half of the inputs (see Figure \ref{fig:SWAP-tree}). With this process, we can write down a formula for $\alpha(h,m)$:
\begin{align}
    \alpha(h,m) = \sum_{k=0}^h \mathbb{P}[k, h-k] \left(\frac{1}{2}\right)^{1 - \delta_{k, h-k}} \alpha(k, m-1) \alpha(h - k, m-1),
\end{align}
where $\delta_{x,h-x}$ is the Kronecker delta which tracks if the amount of orthogonal states in the first and second half are equal. And $\mathbb{P}[k, h-k]$ is the probability that the state has $k$ orthogonal states in the first part and $h-k$ in the second part. For the boundary conditions of $\alpha$ we know that if $h$ = 0, we will not count anything so $\alpha(0,m) = \left(\frac{1}{2}\right)^0 = 1$, and if $h$ = 1, we get a count in every layer so $\alpha(1,m) = \left(\frac{1}{2}\right)^m$. When there are no more divisions we don't count any possible clicks anymore so $\alpha(h,0)=0$. Note that these boundary conditions define a unique $\alpha(h,m)$.

When we randomly permute all the inputs we can compute this probability:
\begin{align}
    \alpha(h,m) &= \sum_{k=0}^h \frac{\binom{2^{m-1}}{k}\binom{2^{m-1}}{h-k}}{\binom{2^m}{h}}  \left(\frac{1}{2}\right)^{1 - \delta_{k, h-k}} \alpha(k, m-1) \alpha(h - k, m-1) \\
    \binom{2^m}{h} \alpha(h,m) &= \sum_{k=0}^h \left(\frac{1}{2}\right)^{1 - \delta_{k, h-k}} \binom{2^{m-1}}{k}\alpha(k, m-1) \binom{2^{m-1}}{h-k}\alpha(h - k, m-1) \\
    \gamma(h,m) &= \sum_{k=0}^h \left(\frac{1}{2}\right)^{1 - \delta_{k, h-k}} \gamma(k,m-1) \gamma(h-k, m-1),
\end{align}
where we have defined $\gamma(h,m) = \binom{2^m}{h} \alpha(h,m)$. Using that the expressions in the sum are equal for $k=i$ and $k=h-i$ we can further simplify the expression. Note if $h$ is odd we have that $k$ is never equal to $h-k$, so we get:
\begin{align}
    \gamma(h,m) &= \sum_{k=0}^h \left(\frac{1}{2}\right) \gamma(k,m-1) \gamma(h-k, m-1) = \sum_{k=0}^{\floor{\frac{h}{2}}} \gamma(k,m-1) \gamma(h-k, m-1).
\end{align}
If $h$ is even we have
\begin{align}
    \gamma(h,m) &= \sum_{k=0}^{\floor{\frac{h}{2}} -1} \gamma(k,m-1) \gamma(h-k, m-1) + \left(\frac{1}{2}\right)^{1 - \delta_{h/2,h/2}} \gamma(h/2, m-1)\gamma(h/2, m-1) \\
    &= \sum_{k=0}^{\floor{\frac{h}{2}}} \gamma(k,m-1) \gamma(h-k, m-1).
\end{align}
We get the same expressions for $h$ odd and even so we can write in general:
\begin{align}
    \gamma(h,m) &= \sum_{k=0}^{\floor{\frac{h}{2}}} \gamma(k,m-1) \gamma(h-k, m-1).
\end{align}
With boundary conditions defined by $\alpha(h,m)$, namely $\gamma(0,m) = \binom{2^m}{0} \alpha(0,m) = 1$, $\gamma(1,m) = \binom{2^m}{1} \alpha(1,m) = 2^m \left(\frac{1}{2}\right)^m = 1$, and $\gamma(h,0) = 0$. Finally, we get the lower bound on the soundness probability
\begin{align}
        \mathbb{P}^{\mathrm{IST}}_{s}((n,h)) = 1 - \mathbb{E}\left[\left( \frac{1}{2} \right)^{\#\text{poss. clicks}(h,m)}\right] \geq 1 - \alpha(h, \log_2(n)) = 1 - \frac{\gamma(h, \log_2(n))}{\binom{n}{h}},
\end{align}
as desired.
\end{proof}

Note that $\gamma(h,\log_2(n))$ is increasing in $h$ up until $n/2$, as well as increasing in $n$, so having more orthogonal states instead of the other states will give us better bounds on the soundness, but also having more states in general will always give better bounds.

The proof above assumes a partition into two parts $\mu = (n-h, h)$. For a general partition $\mu$ a way to lower bound the number of possible clicks can be done as follows. Join the partitions together such that we get two parts both of maximal size, then we can apply the theorem again and since $\gamma(h,\log_2(n))$ is increasing in $h$ we take it as close as possible to $n/2$. Then regard all the states in one partition as the same (i.e. the $0$ and $1$), and apply our counting strategy to get a lower bound on the number of possible clicks. To see that this again undercounts, note that seeing two states as the same states does not increase the number of possible clicks, and only does not count possible clicks between partitions that get combined. Then we get the following corollary:

\begin{corollary}
    The soundness probability of the \hyperref[alg:IST]{IST} for $n=2^m$ and some partition $\mu$ is lower bounded by,
    \begin{align}
        \mathbb{P}^{\mathrm{IST}}_s(\mu) \geq 1 - \frac{\gamma(h^*,\log_2(n))}{\binom{n}{h^*}},
    \end{align}
    where $h^* = \sum_{i=k}^h \mu_i$, where $k = \min\{l \, | \, \sum_{i=l}^h \mu_i \leq \frac{n}{2}\}$.
\end{corollary}

\begin{lemma}
    Fix $n=2^m$ and an integer $h \geq 1$. Then:
    \begin{enumerate}
        \item for every $h \geq 2$, $\gamma(h,m) = 0$ for every integer $m < \log_2(h)$,
        \item for any integers $h,m$ we have a symmetry $\gamma(h,m) = \gamma(n - h,m)$,
        \item $\gamma(h,m)$ is a polynomial in $m$ of degree $h-1$. 
    \end{enumerate} 
    \label{lem:polynomial}
\end{lemma}
\begin{proof}
    (1) and (2) are obvious from the definition of $\gamma(h,m)$. Let's prove (3). We proceed by induction in $h$. The base of induction $h=1$ is clear: $\gamma(1,m) = 1$ for every $m$ from definition. Now assume $h \geq 2$ and $\gamma(h-1,m) = p_{h-2}(m)$, where $p_{h-2}(m)$ is a polynomial of degree $h-2$.  Then using the defining relation of $\gamma(h,m)$ we get
    \begin{align}
        \gamma(h,m) &= \gamma(h,m-1) + \sum_{k=1}^{\floor{\frac{h}{2}}} \gamma(k,m-1) \gamma(h-k,m-1)
        \\
        &= \gamma(h,m-2) + \sum_{k=1}^{\floor{\frac{h}{2}}} \sum_{x=m-2}^{m-1} \gamma(k,x) \gamma(h-k,x) \\
        &= \sum_{k=1}^{\floor{\frac{h}{2}}} \sum_{x=0}^{m-1} \gamma(k,x) \gamma(h-k,x) = \sum_{k=1}^{\floor{\frac{h}{2}}} \sum_{x=0}^{m-1} p_{k-1}(x)p_{h-k-1}(x) = \sum_{k=1}^{\floor{\frac{h}{2}}} \sum_{x=0}^{m-1} q_{h-2}(x),
    \end{align}
    where $q_{h-2}(x)$ is some polynomial of degree $h-2$. But
    \begin{equation}
        \sum_{x=1}^{m-1} q_{h-2}(x) = p_{h-1}(m-1),
    \end{equation}
    for some polynomial $p_{h-1}(m-1)$ of degree $h-1$ due to the well known Bernoulli's formula
    \begin{equation}
        \sum_{x=1}^{m-1} x^{h-2} = \frac{(m-1)^{h-1}}{h-1} + g(m-1),
    \end{equation}
    where $g(m-1)$ is a polynomial of degree $h-2$. Therefore we have
    \begin{equation}
         \gamma(h,m) = \sum_{k=1}^{\floor{\frac{h}{2}}} p_{h-1}(m-1),
    \end{equation}
    so $\gamma(h,m)$ is a polynomial of degree $h-1$.
\end{proof}

We will now show that for any fixed $h$, for sufficiently large $m$ the \hyperref[alg:IST]{IST} will achieve optimal soundness error.

\begin{corollary} Let $n = 2^m$ for $m \in \mathbb{N}$ and fix any $1 \leq h \leq n/2 $ Then there exists a $n^*_0 \in \mathbb{N}$ such that the \hyperref[alg:IST]{IST} protocol achieves perfect completeness and for all $n \geq n^*_0$ has soundness probability $\geq 1- 1/n$.
\label{cor:IST}
\end{corollary}
\begin{proof}
Perfect completeness follows from Theorem~\ref{thm:swap_trees}. Since we are interested in the worst-case error in the soundness setting, we can assume $\mu = (n-h,h)$.  From Lemma~\ref{lem:polynomial}, we know that $\gamma(h,m)$ is a polynomial in $m$ of degree $h-1$. Using that $m= \log_2 n$, we have by Theorem~\ref{thm:swap_trees} that the soundness probability, given a fixed $h$, can be lower bounded as
\begin{equation}
        \mathbb{P}^{\mathrm{IST}}_{s}(n,h) \geq 1 - \frac{\gamma(h,\log_2(n))}{\binom{n}{h}} \geq
        1- \frac{ c_h   (\log_2 n )^{h-1} }{\left(\frac{n}{h}\right)^h} = 1- \frac{1}{n} \cdot f(n,h),
\end{equation}
where $f(n,h)$ is given by
\begin{align*}
    f(n,h) \defeq c_h h \left(\frac{h \log_2 n}{n}\right)^{h-1},
\end{align*}
with $c_h$ is a constant that only depends on $h$.  Note that we have $c_1 = 1$ and $f(n,1)=1$. For any $2 \leq h  \leq n/2$, we have that there exists a constant $n_0$ which only depends on $c_h$ and $h$ such that for all $n \geq n_0$ $f(n,h) \leq 1$. Hence, taking $n^*_0$ to be the maximum over all such $n_0$ for all $h$, the statement directly follows.
\end{proof}

\section{Discussion and open problems}
In the definition of the Quantum State Identity problem we considered the case in which the input states are \textit{exactly} orthogonal to each other.
One could relax this condition to states that are approximately orthogonal, i.e. we get the promise that for some $\varepsilon>0$ the inner product between two states $\ket{\psi_i}$ and $\ket{\psi_j}$ is either $|\!\bra{\psi_i} \ket{\psi_j}\!| \geq 1-\varepsilon$ or  $|\!\bra{\psi_i} \ket{\psi_j}\!| \leq \varepsilon$.
It would be natural to still perform the permutation test in this approximate setting and analyze its performance.
A first question is whether the permutation test remains optimal under this relaxed promise.
A second question is whether there exist simpler operations that achieve the same performance as the permutation test, which could imply that in some practical settings one can attain optimal performance with an easier experimentally implementable measurement.

The $G$-test we propose gives an analytical description of the performance of a protocol that applies the permutation test to an arbitrary subgroup $G$.
To better understand the performance one would need to compute or get bounds on the Kostka numbers and the multiplicities of the trivial irrep of this subgroup $G$ inside irreps of the symmetric group $\S_n$.
A further open question is to determine, for a given subgroup $G$ and partition $\mu$, the threshold prior $p^*(G,\mu)$ below which the trivial test outperforms the $G$-test, analogous to the threshold $p^*(\mu)$ established for the permutation test in Theorem~\ref{thm:permtest_trivial_optimal}.

We introduce the Iterated Swap Tree as a new simple protocol that answers the $\hyperref[def:QSI]{\QSI_\mu^{p}}$ problem.
We derive lower bounds for the number of possible clicks which enables us to give bounds on its performance.
The quantity $\gamma(h,m)$ that gives this bound is defined by a recurrence relation.
It would be of interest to derive an exact formula for $\gamma(h,m)$ or get some exact formula $f(h,m)$ that bounds this recurrence relation, to get a better understanding of the performance of the Iterated Swap Tree in the case of non-fixed $h$.
It is also an interesting open question whether the bound on the soundness we derive is tight, or whether the true performance of the Iterated Swap Tree is in fact closer to that of the optimal permutation test.

The proposed Iterated Swap Tree can be seen as a $G$-test where $G$ is the iterated wreath product of the symmetric groups of size 2 \cite{orellana2004rooted}.
This notion can be generalized to symmetric groups of larger prime size $p$, such that inputs of $p^m$ qudits can also be analyzed.
One can then take the wreath product between different-sized groups, which would give the tools to apply a similar strategy as in the Iterated Swap Tree to arbitrarily sized inputs.
Namely, one can take the prime factorization of the number of qudits one receives as input, and then take the iterated wreath product between the symmetric groups whose sizes are these prime factors.

Finally, one can generalize the $\hyperref[def:QSI]{\QSI}$ problem to a more fine-grained property testing problem: given $\rho^{\mu}_{\neq}$ for some $\mu \pt n$ ($\mu_2$ could be zero here), what is the most efficient protocol to reconstruct $\mu$?
This problem for qubits was studied in \cite{montanaro2009symmetric}.
It was shown that the measurement called \emph{weak Schur sampling} \cite{childs2007weak} is optimal for the qubit case.
We conjecture that the qudit generalization of weak Schur sampling would also solve the general qudit problem optimally, since the Schur-Weyl duality structure underlying the optimality argument naturally extends to the higher-dimensional setting.

\section*{Acknowledgements}
DG was supported by an NWO Vidi grant (Project No VI.Vidi.192.109). PVL and HB were supported by the Dutch Research Council (NWO/OCW), as part of the NWO Gravitation Programme Networks (project number 024.002.003). JW was supported by the Dutch Ministry of Economic Affairs and Climate Policy (EZK), as part of the Quantum Delta NL Programme.

\printbibliography

\end{document}